\documentclass[10pt,journal]{IEEEtran}
\usepackage{amsfonts}
\usepackage{amssymb,epic,eepic}
\usepackage{amsmath}
\newcommand{\hr}{{\cal H}}

\newcommand{\cc}{{\mathbb C}}
\newcommand{\rr}{{\mathbb R}}
\newcommand{\nn}{{\mathbb N}}
\newcommand{\idn}{\mathbf{1}}

\newcommand{\eps}{{\varepsilon}}        

\newtheorem{theorem}{Theorem}[section]         
\newtheorem{lemma}[theorem]{Lemma}             
\newtheorem{remark}[theorem]{Remark}           

\begin{document}

\title{Classical Capacities of Compound and Averaged Quantum Channels}

\author{Igor Bjelakovi\'c and Holger Boche\\[2mm]
\small Heinrich-Hertz-Chair for Mobile Communications\\ Technische Universit\"at Berlin\\Werner-von-Siemens-Bau (HFT 6), Einsteinufer 25, 10587 Berlin, Germany\\
\&\\
Institut f\"ur Mathematik, Technische Universit\"at Berlin\\
Stra\ss e des 17. Juni 136, 10623 Berlin, Germany\\
Email: \{igor.bjelakovic, holger.boche\}@mk.tu-berlin.de 
\thanks{This work is supported by the Deutsche Forschungsgemeinschaft DFG via project Bj 57/1-1 ''Entropie und Kodierung gro\ss er Quanten-Informationssysteme''.}}
\maketitle
\begin{abstract}We determine the capacity of compound classical-quantum channels. As a consequence we obtain the capacity formula for the averaged classical-quantum channels. The capacity result for compound channels demonstrates, as in the classical setting, the existence of reliable universal classical-quantum codes in scenarios where the only a priori information about the channel used for the transmission of information is that it belongs to a given set of memoryless classical-quantum channels.
Our approach is based on a universal classical approximation of the quantum relative entropy which in turn relies on a universal hypothesis testing result.
\end{abstract}

\section{Introduction}
 In this paper we present the coding theorems for compound and averaged channels with classical input and quantum output (cq-channels). The result nicely supplements recent results of Datta and Dorlas \cite{datta-dorlas} where they considered finite weighted sums of memoryless quantum channels and determined their classical capacity. This is one of the basic examples of channels with long-term memory. This is obviously equivalent to the determination of the classical capacity for the associated compound channel consisting of finitely many channels, since for finite sums we can easily bound the error probabilities of the individual memoryless branches by the error probability of the averaged channel and vice versa. Unfortunately, the beautiful method of proof in \cite{datta-dorlas} does not apply when the number of channels is infinite.\\
Roughly, the interest in compound channels is motivated by the fact that in many situations we have only a limited knowledge about the channel which is used for the transmission of information. In the compound setting we know merely that the memoryless cq-channel which is in use belongs to some given finite or infinite set of memoryless cq-channels which is a priori known to the sender and receiver. Their goal is to construct coding-decoding strategies that work well for the whole set of channels simultaneously. The situation is comparable with the universal source coding scenario considered in \cite{jozsa-horodecki} by Jozsa and M., P., and R. Horodecki. Averaged cq-channels are close relatives of compound channels, the difference being that in this situation the communicating parties have access to an additional a priori probability distribution governing the appearance of the particular member of the compound channel.\\
 The paper is organized as follows: In Section \ref{rev-class-comp} we give a rapid overview of the classical theory of compound channels. Whereas Section \ref{def-cq-comp} is devoted to the notion of compound cq-channels and the definition of the capacity for this class of channels. The subsequent Section \ref{universal-approximation} contains the first pillar of our argument. Namely, we construct, using an idea going back to Nagaoka, a universal classical approximation of the quantum relative entropy for classes of uncorrelated quantum states. The central Section \ref{compound} starts with a relation between a minimization procedure arising in universal hypothesis testing and the minimization process required for the determination of the capacity of compound cq-channels which is based on Donald's inequality (cf. Lemmata \ref{donald} and \ref{equality-2}). Then we proceed with the direct and the (strong) converse part of the coding theorem for compound cq-channels\footnote{After the submission of this paper Hayashi \cite{hayashi-compound} obtained a similar result via Weyl-Shur duality. His result can be used to give another proof  of the direct part of the coding theorem for averaged channels. His error bounds are exponenial but depend on the channel.}.
As a by-product we can prove in Section \ref{av-channel} the coding theorem and the weak converse for arbitrary averaged cq-channels with memoryless branches. This extends, in part, the results of Ahlswede \cite{ahlswede} to the cq-situation. Moreover, the results of Datta and Dorlas \cite{datta-dorlas} are generalized to averages of memoryless cq-channels with respect to arbitrary probability measures, provided the set of channels has some appropriate measurable structure.
\subsection{Notation}
We will assume tacitly throughout the paper that all Hilbert spaces are over the field $\cc$. The identity operator acting on a Hilbert space $\hr$ is denoted by $\idn_{\hr}$ or simply by $\idn$ if it is clear from the context which Hilbert space is under consideration. The set of density operators acting on the finite-dimensional Hilbert space $\hr$ is denoted by $\mathcal{S}(\hr)$ and the set of probability distributions on a finite set $A$  will be abbreviated by $\mathcal{P}(A)$. $|A|$ denotes the cardinality of the set $A$. The projection onto the range of a density operator $\rho\in \mathcal{S}(\hr)$, $\dim \hr <\infty$, is called the support of $\rho$ and we dedicate the notation $\mathrm{supp}(\rho)$ to it.\\
The relative entropy of the state (i.e. density operator) $\rho$ with respect to the state $\sigma$ is given by
\[  S(\rho||\sigma):=\left\{ \begin{array}{ll}  
\textrm{tr}(\rho\log \rho-\rho\log \sigma) & \textrm{if } \mathrm{supp}(\rho)\le \mathrm{supp}(\sigma)\\
\infty & \textrm{else}
\end{array}\right.,
 \] 
where $\textrm{tr}$ stands for the trace and $\log$ is the binary logarithm. The classical analog of the relative entropy known as Kullback-Leibler distance is defined by
\[  D(p||q):=\left\{ \begin{array}{ll}  
\sum_{a\in A}p(a)\log p(a)-p(a)\log q(a) & \textrm{if } p\ll q\\
\infty & \textrm{else}
\end{array}\right.,
 \] 
where $p,q\in\mathcal{P}(A)$. The relation $p\ll q$ means that $q(a)=0$ for some $a\in A$ implies $p(a)=0$ or, equivalently, that $\mathrm{supp}(p)\subset \mathrm{supp}(q)$, where $\mathrm{supp}(p):=\{a\in A:p(a)>0  \}$.\\
Von Neumann entropy of a density operator $\rho\in\mathcal{S}(\hr)$, $\dim \hr <\infty$, is defined to be $S(\rho):=-\textrm{tr}(\rho\log\rho)$. The Shannon entropy of $p\in \mathcal{P}(A)$, $|A|<\infty$, is given by $H(p):=-\sum_{x\in A}p(x)\log p(x)$.\\
The $n$-fold Cartesian product of a finite set $A$ with itself is denoted by $A^n$. We set $x^n:=(x_1,\ldots ,x_n)$ for sequences $(x_1,\ldots ,x_n)\in A^n$.\\
Notation we use for the logarithms is as follows: $\log_{a}$ is the logarithm to the base $a>1$ and $\log$ is understood as $\log_2$.
\section{Short Overview of the Classical Theory of Compound Channels}\label{rev-class-comp}
The basic classical theory of compound channels was developed independently by Blackwell, Breiman, Thomasian \cite{blackwell} and Wolfowitz \cite{wolfowitz-compound}. Blackwell, Breiman and Thomasian proved the coding theorem with the weak converse. Wolfowitz, on the other hand, obtained the coding theorem with the strong converse for the maximum error criterion by an entirely different method of proof. We recall at this place briefly the capacity formula just to emphasize the similarity to the capacity formula (\ref{quant-compound-capacity}) for the cq-case.\\
For an arbitrary set $T$ and finite sets $A$, $B$ we consider the family of discrete channels $W_t:A\to B$, $t\in T$. The compound channel, denoted by $T$, is simply the whole family of discrete memoryless channels $\{W_{t}^{n}  \}_{t\in T,n\in \nn}$.\\
Let $\lambda\in (0,1)$. An $(n,M_n,\lambda)_{\max}$-code for the compound channel $T$ is set of tuples $(x^n(i), B_i)_{i=1}^{M_n}$ where $x^{n}(i)\in A^n$, $B_i\subseteq B^n$, $B_i \cap B_j=\emptyset$ for $i\neq j$ and
\[ W_t^n(B_i|x^n(i))\ge 1-\lambda \]
for all $i=1,\ldots , M_n$ and all $t\in T$. A similar definition of the $(n, M_n,\lambda)_{\mathrm{av}}$-codes can be given simply by replacing the maximum error criterion by the average one. Thus the goal is to find reliable codes which work well for all discrete memoryless channels indexed by the set $T$.\\
The work \cite{blackwell}, \cite{wolfowitz-compound} can be summarized as follows: The weak capacity of the compound channel $T$ with respect to both the maximum and average error criteria is given by
\begin{equation}\label{class-compound-capacity}
  C(T)=\max_{p\in \mathcal{P}(A)}\inf_{t\in T}I(p,W_t),
\end{equation}
where $\mathcal{P}(A)$ denotes the set of probability distributions on $A$ and $I(p,W_t)$ is the mutual information of the channel $W_t$ with respect to the input distribution $p$. Wolfowitz has shown that the RHS of (\ref{class-compound-capacity}) is  the strong capacity with respect to the maximum error criterion. Ahlswede gives an example in \cite{ahlswede-compound} that demonstrates that, surprisingly, the strong converse need not hold for compound channels if the average probability of error is used in the definition of the capacity.
\section{Compound CQ- Channels}\label{def-cq-comp}
We consider here a set of cq-channels $W_t:A\ni x\mapsto D_{t,x}\in \mathcal{S}(\hr)$, $t\in T$, for an arbitrary set $T$ where $A$ is a finite set and $\hr$ is a finite-dimensional Hilbert space. The $n$-th memoryless extension of the cq-channel $W_t$ is given by $W_{t}^{n}(x^n):=D_{t,x^n}:=D_{t,x_1}\otimes \ldots \otimes D_{t,x_n}$ for $x^n \in A^n$.\\
The \emph{compound cq-channel} is given by the family $\{W_t^n\}_{t\in T,n\in \nn}$. We will write simply $T$ for the compound cq-channel. \\
An $n$-\emph{code}, $n\in \nn$, for the compound cq-channel $T$ is a family $\mathcal{C}_n:=(x^n(i),b_i)_{i=1}^{M_n}$ consisting of sequences $x^n(i)\in A^n$ and positive semi-definite operators $b_i\in \mathcal{B}(\hr)^{\otimes n}$ such that $\sum_{i=1}^{M_n}b_i\le \idn^{\otimes n} $. The number $M_n$ is called the \emph{size} of the code.\\
A code $\mathcal{C}_n$ is called a $(n,M_n,\lambda)_{\max}$-code for the compound cq-channel $T$ if the size of $\mathcal{C}_n$ is $M_n$, $x^n(i)\in A^n$ and if
\begin{equation}\label{error-def}
  e_m (t,\mathcal{C}_n):=\max_{i=1,\ldots ,M_n}(1-\textrm{tr}(D_{t,x^n(i)}b_i))\le \lambda \qquad \forall t\in T.
\end{equation}
with an analog definition of an $(n,M_n,\lambda)_{\textrm{av}}$-code w.r.t average error probability criterion, i.e. we replace $e_m (t,\mathcal{C}_n)\le \lambda $ by
\[  e_a (t,\mathcal{C}_n):=\frac{1}{M_n}\sum_{i=1}^{M_n}(1-\textrm{tr}(D_{t,x^n(i)}b_i))\le \lambda \qquad \forall t\in T \]
in the definition.\\
 Thus an $(n,M_n,\lambda)_{\max}$-code for the compound channel $T$ ensures that the maximal error probability for all channels of class $T$ is bounded from above by $\lambda$. A more intuitive description of the compound channel is that the sender and receiver actually don't know which channel from the set $T$ is used during the transmission of the $n$-block. Their prior knowledge is merely that the channel is memoryless and belongs to the set $T$.
This is a channel analog of the universal source coding problem for a set of memoryless sources (cf. \cite{jozsa-horodecki}).\\
A real number $R\ge 0$ is said to be an \emph{achievable rate} for the compound channel if there is a sequence of codes $(\mathcal{C}_n)_{n\in\nn}$ of sizes $M_n$ such that
\begin{equation}\label{achiev-rate-def}
 \liminf_{n\to\infty}\frac{1}{n}\log M_n\ge R, 
\end{equation}
and
\begin{equation}
 \lim_{n\to\infty}\sup_{t\in T}e(t, \mathcal{C}_n)=0. 
\end{equation}
The \emph{weak capacity}, denoted by $C(T)$, of the compound channel $T$ is defined as the least upper bound of all achievable rates.\\
$R\ge 0$ is called a \emph{$\lambda$-achievable rate} for the compound channel $T$, $\lambda\in[0,1)$, if there is a sequence of codes $(\mathcal{C}_n)_{n\in\nn}$ of sizes $M_n$ for which (\ref{achiev-rate-def}) holds but the error condition is relaxed to
\[  \sup_{t\in T}e(t, \mathcal{C}_n)\le \lambda \qquad \forall n\in\nn.\]
The \emph{$\lambda$-capacity} $C(T,\lambda)$ is the least upper bound of all $\lambda$-achievable rates. \\
The Holevo information of a cq-channel $W_t: A\to\mathcal{S}(\hr)$ with respect to the input distribution $p\in \mathcal{P}(A)$ is defined by
\begin{equation}\label{holevo-inf-def}
  \chi (p,W_t):=S(D_t)-\sum_{x\in A}p(x)S(D_{t,x})
\end{equation}
where $S(\cdot)$ stands for von Neumann entropy.\\
As shown in \cite{holevo}, \cite{schumacher}, \cite{winter}, and \cite{ogawa-nagaoka} the $\lambda$-capacity of a single memoryless cq-channel $W$ is given by
\[C(W,\lambda)=\max_{p\in \mathcal{P}(A)}\chi(p,W)\quad \forall \lambda\in (0,1).\]
The main result of our paper is an analog of the capacity formula (\ref{class-compound-capacity}) and can be stated as follows.
\begin{theorem}
Let $T$ be an arbitrary compound cq-channel with finite input alphabet $A$ and finite-dimensional output Hilbert space $\hr$. Then 
\begin{equation}\label{quant-compound-capacity}
  C(T,\lambda)=\max_{p\in \mathcal{P}(A)}\inf_{t\in T}\chi(p,W_t)
\end{equation}
holds for any $\lambda\in (0,1)$.
\end{theorem}
\begin{proof} The achievability, i.e. the inequality
\[ C(T,\lambda)\ge\max_{p\in \mathcal{P}(A)}\inf_{t\in T}\chi(p,W_t)  \]
follows from Theorem \ref{comp-direct-part}. On the other hand, Theorem \ref{comp-strong-converse} shows that we cannot be better than the right hand side of (\ref{quant-compound-capacity}) which establishes the inequality
\[C(T,\lambda)\le\max_{p\in \mathcal{P}(A)}\inf_{t\in T}\chi(p,W_t).  \]
\end{proof}
\section{Universal Classical Approximation of the Quantum Relative Entropy}\label{universal-approximation}
The purpose of this section is the derivation of a universal classical approximation of quantum relative entropies of a given set $\Omega\subset \mathcal{S}(\hr)$ with respect to a reference state $\sigma\in \mathcal{S}(\hr)$. The first result of this kind was obtained in the paper \cite{hiai-petz} by Hiai and Petz in the case $|\Omega|=1$. Basically they have shown that for given states $\rho,\sigma\in \mathcal{S}(\hr)$ we can approximate $S(\rho^{\otimes l}||\sigma^{\otimes l})$ by the Kullback-Leibler divergence of the probability distributions $p_l$ and $q_l$ given by
\[ p_l(i)=\textrm{tr}(\rho^{\otimes l}P_i), \quad q_l(i)=\textrm{tr}(\sigma^{\otimes l}P_i ), \]
for suitable projections $P_i=P_i(l,\rho,\sigma)\in \mathcal{B}(\hr)^{\otimes l}$ with $\sum_{i=1}^{N_l} P_i=\idn_{\hr}^{\otimes l}$. The approximation error does not exceed $\dim \hr \cdot \log (l+1)$. Precisely, Hiai and Petz have shown that
\[S(\rho^{\otimes l}|| \sigma^{\otimes l})\ge D(p_l||q_l)\ge S(\rho^{\otimes l}|| \sigma^{\otimes l})-\dim \hr \cdot \log (l+1).  \]
This approximation result for quantum relative entropy was the crucial step for a construction of projections $Q_n\in\mathcal{B}(\hr)^{\otimes n}$ for each $n\in\nn$ with the properties
\begin{enumerate}
\item $\lim_{n\to\infty}\textrm{tr}(\rho^{\otimes n}Q_n)=1$ and,
\item $\limsup_{n\to\infty}\frac{1}{n}\log \textrm{tr}(\sigma^{\otimes n}Q_n)\le S(\rho||\sigma)$.
\end{enumerate}
These properties are exactly the direct part of the quantum version of Stein's Lemma. Subsequently, Nagaoka observed that these arguments can be reversed, i.e. starting from the direct part of Stein's Lemma we can construct a classical approximation of the quantum relative entropy by simply considering the projections $Q_n$ and $\idn_{\hr}^{\otimes n}-Q_n$ and probability distributions $p_n=(\textrm{tr}(\rho^{\otimes n}Q_n), 1- \textrm{tr}(\rho^{\otimes n}Q_n) )$, $q_n=(\textrm{tr}(\sigma^{\otimes n}Q_n), 1- \textrm{tr}(\sigma^{\otimes n}Q_n)   )$\footnote{We learned this from the paper \cite{ogawa-hayashi} by Ogawa and Hayashi who attribute this observation to Nagaoka.} (cf. our inequality chain (\ref{nagaoka-like-calculation}) for more details). It is an interesting fact that Nagaoka's argument produces for each $n\in\nn$ pairs of projections which give rise to a good approximation of the quantum relative entropy. \\
Our approach to the universal classical approximation is motivated by Nagaoka's argument and therefore we need a universal version of Stein's Lemma or Sanov's Theorem from \cite{q-sanov}. Actually we need a slightly sharper result than that obtained in \cite{q-sanov}. The main tool to obtain this sharpening is contained in the following
\begin{lemma}\label{c-sharp-sanov}
Let $X$ be a finite set and $r\in \mathcal{P}(X)$ with $r(x)>0$ for all $x\in X$. Then for each $\delta>0$, $k\in \nn$, and any set $\Omega_k\subset \mathcal{P}(X)$ there is a subset $X_{k,\delta}\subset X^k$ with
\begin{enumerate}
\item $q^{\otimes k}(X_{k,\delta})\ge 1- (k+1)^{|X|}2^{-kc\delta^2}$ for all $q\in \Omega_k$ with a universal constant $c>0$.
\item 
\[r^{\otimes k}(X_{k,\delta})\le (k+1)^{|X|}2^{-k(D(\Omega_k||r)-\eta(\delta,r) )},  \]
with $D(\Omega_k||r):=\inf_{q\in \Omega_k}D(q||r)$
and $\eta (\delta,r):=-\delta\log \frac{\delta}{|X|}-\delta\log r_{\min} $, where $r_{\min}$ denotes the smallest positive value of $r$.
\end{enumerate}
\end{lemma}
\begin{proof}
The proof uses the well known type bounding techniques from \cite{csiszar} and \cite{shields} and is therefore omitted.
\end{proof}
A (discrete) projection valued measure (PVM) on a finite dimensional Hilbert space $\cal{K}$ is a set $\mathcal{M} :=\{P_i\}_{i=1}^{m}$ consisting of projections $P_i\in\mathcal{B}(\cal{K})$ such that $\sum_{i=1}^{m}P_i=\idn_{\cal{K}}$. For two states $\rho,\sigma\in\mathcal{S}(\cal{K})$ and any  PVM $\mathcal{M}$ on $\cal{K}$ we define
\[S_{\mathcal{M}}(\rho||\sigma):=\sum_{i=1}^{m}\textrm{tr}(\rho P_i)\log \textrm{tr}(\rho P_i)-\textrm{tr}(\rho P_i)\log \textrm{tr}(\sigma P_i)  \]
if $(\textrm{tr}(\rho P_i))_{i=1}^{m}\ll (\textrm{tr}(\sigma P_i) )_{i=1}^{m}$ and 
\[S_{\mathcal{M}}(\rho||\sigma):=\infty  \]
else.
\begin{theorem}\label{abelian-approximation}
Let $\sigma\in\mathcal{S}(\hr)$ be invertible. Then for each $l\in \nn$ there is a real number $\zeta_l (\sigma)$ with $\lim_{l\to\infty}\zeta_l(\sigma)=0$ such that for any set $\Omega_l\subset \mathcal{S}(\hr)$ there is a PVM $\mathcal{M}_l=\{P_l, \idn_{\hr}^{\otimes l}-P_l\}$ on $\hr^{\otimes l}$ with
\[S_{\mathcal{M}_l}(\rho^{\otimes l}||\sigma^{\otimes l})\ge l(S(\Omega_l||\sigma)-\zeta_l(\sigma)  )  \]
for all $\rho\in \Omega_l$ with $S(\Omega_l||\sigma):=\inf_{\rho\in\Omega_l}S(\rho||\sigma)$. Consequently,
\[ \inf_{\rho\in\Omega_l}S_{\mathcal{M}_l}(\rho^{\otimes l}||\sigma^{\otimes l})\ge l(S(\Omega_l||\sigma)-\zeta_l (\sigma) ). \]
\end{theorem}
\begin{proof}
The proof is based on the following observation: Let $\mathcal{M}_l=\{P_l, \idn_{\hr}^{\otimes l}-P_l\}$ be any PVM on $\hr^{\otimes l}$ with the properties
\begin{enumerate}
\item $\textrm{tr}(\rho^{\otimes l} P_l)\ge 1-\tau_{1,l}$ for all $\rho\in\Omega_l$ with $\lim_{l\to\infty}\tau_{1,l}=0$ and
\item $\textrm{tr}(\sigma^{\otimes l}P_l)\le 2^{-l(S(\Omega_l||\sigma)- \tau_{2,l} )} $ with $\lim_{l\to\infty}\tau_{2,l}=0.$
\end{enumerate}
Then using these relations we can lower-bound $S_{\mathcal{M}_l}$ for each $\rho\in\Omega_l$ as follows: First of all, since $\sigma$ is invertible we have
\[S(\rho^{\otimes l}||\sigma^{\otimes l})< \infty  \]
for each $\rho\in\Omega_l$. Thus, the monotonicity of the relative entropy yields
\[  S_{\mathcal{M}_l}(\rho^{\otimes l}||\sigma^{\otimes l})\le S(\rho^{\otimes l}||\sigma^{\otimes l})< \infty \]
for all $\rho\in\Omega_l$. Consequently we can lower-bound $\frac{1}{l} S_{\mathcal{M}_l}(\rho^{\otimes l}||\sigma^{\otimes l}) $ using the relations 1) and 2):
\begin{eqnarray}\label{nagaoka-like-calculation}
 \frac{1}{l} S_{\mathcal{M}_l}(\rho^{\otimes l}||\sigma^{\otimes l})&\ge& -\frac{1}{l}H((\textrm{tr}(\rho^{\otimes l}P_l,\textrm{tr}(\rho^{\otimes l}(\idn_{\hr}^{\otimes l}-P_l) )  )\nonumber\\
&& -\textrm{tr}(\rho^{\otimes l}P_l)\frac{1}{l}\log \textrm{tr}(\sigma^{\otimes l}P_l)\nonumber\\
&\ge& -\frac{\log 2}{l} + \textrm{tr}(\rho^{\otimes l}P_l)(S(\Omega_l||\sigma)- \tau_{2,l} )\nonumber\\
&\ge&  -\frac{1}{l}+(1-\tau_{1,l})(S(\Omega_l||\sigma)- \tau_{2,l} )\nonumber\\
&\ge & S(\Omega_l||\sigma)-\zeta_l(\sigma),
\end{eqnarray}
with 
\begin{equation}\label{abel-approx-1}
\zeta_l(\sigma):= (1-\tau_{1,l})\tau_{2,l}-\tau_{1,l}\log \lambda_{\min}(\sigma) +\frac{1}{l},
\end{equation}
where $\lambda_{\min}(\sigma)$ denotes the smallest eigenvalue of $\sigma$.\\
Thus our remaining job is the construction of the PVM with the properties described above. To this end let $l\in \nn$ and $\Omega_l\subset \mathcal{S}(\hr)$ be given. For $m\in\nn$ we can find $k,y\in \nn$ with $0\le y<m$ such that $l=km +y$. Then applying exactly the same bounding technique as in the proof of Theorem 2 in \cite{q-sanov} but using our Lemma \ref{c-sharp-sanov} instead of their Lemma 1 we obtain for each $\delta>0$ a projection $P_{l,\delta}\in \mathcal{B}(\hr)^{\otimes l}$ with
\begin{enumerate}
\item $\textrm{tr}(\rho^{\otimes l}P_{l,\delta})\ge 1-(k+1)^{d^{m}}2^{-kc\delta^2}$ with a universal constant $c>0$ and where $d=\textrm{dim}(\hr)$,
\item 
\begin{eqnarray*}
 \frac{1}{l}\log \textrm{tr}(\sigma^{\otimes l}P_{l,\delta})&\le & -S(\Omega_l||\sigma)+d \frac{\log (m+1)}{m}\\    
&& + (d^{2m}+d^m)\frac{\log (k+1)}{km}\\
&& +\eta(\delta,\sigma),
\end{eqnarray*} 
with
\[\eta(\delta,\sigma)=-\delta\log \frac{\delta}{d}-\delta \log \lambda_{\min}(\sigma).\]
\end{enumerate}
Choosing $m=m_l:=\lceil \log_d (l^{1/8} )\rceil$ it is easily seen that for $k=k_l=\frac{l-y_l}{m_l}$ with $0\le y_l < m_l$ and $\delta_l:=l^{-1/4}$ we have
\[\lim_{l\to\infty} \tau_{1,l} =0 \textrm{ and } \lim_{l\to\infty}\tau_{2,l} =0, \]
where
\begin{equation}\label{abel-approx-2}
\tau_{1,l}:= (k_l+1)^{d^{m_l}}2^{-k_lc\delta_l ^2},  
\end{equation}
and 
\begin{equation}\label{abel-approx-3}
\tau_{2,l}:= d \frac{\log (m_l+1)}{m_l}   
 + (d^{2m_{l}}+d^{m_l})\frac{\log (k_l+1)}{k_l m_l} +\eta(\delta_l,\sigma).  
\end{equation}
The desired PVM is then given by $\mathcal{M}_l:=\{P_l,\idn_{\hr}^{\otimes l}-P_l\}$ with $P_l:=P_{l,\delta_l}$. 
\end{proof}
\begin{remark}
An alternative proof of Theorem \ref{abelian-approximation} might be based on the techniques developed by Hayashi in \cite{hayashi0}, \cite{hayashi}. He constructs there a sequence of PVM's on $\hr^{\otimes l}$ via representation theory of Lie groups which depends merely on $\sigma$ and shows how to derive Stein's Lemma. Thus we are forced to uniformly bound the errors of the first and second kind in Hayashi's setting for the whole family $\Omega_l$ in order to obtain a universal abelian approximation of the quantum relative entropy.
\end{remark}
\section{Capacity of Compound CQ-Channels}\label{compound}
Let $T$ be an arbitrary compound channel and  for a fixed $p\in \mathcal{P}(A)$ define
\[\Omega_p:=\left\{\rho_t:=\sum_{x\in A}p(x)|x\rangle\langle x |\otimes D_{t,x}:t\in T\right\},  \]
where each $\rho_t\in \Omega_p$ is seen as a density operator in $\mathcal{A}_{\textrm{diag}}\otimes \mathcal{B}(\hr)$ with
\[ \mathcal{A}_{\textrm{diag}}:=\bigoplus_{x\in A}\cc|x\rangle\langle x| \]
 being the algebra of operators diagonal w.r.t. the basis $\{|x\rangle  \}_{x\in A}$ of $\cc^{|A|}$\footnote{ $\mathcal{A}_{\textrm{diag}}$ has a natural structure of a $^{\ast}$-algebra, thus $\mathcal{A}_{\textrm{diag}}\otimes \mathcal{B}(\hr) $ is an admissible construction.}. 
Moreover, for each $t\in T$ we set
\[ \sigma_t:=\sum_{x\in A}p(x)D_{t,x}. \]
In what follows we identify the probability distribution $p$ with a diagonal density operator, i.e. we set
\[p=\sum_{x\in A}p(x)|x\rangle \langle x|\in  \mathcal{A}_{\textrm{diag}} .\]
It is well known that
\[S(\rho_t||p\otimes \sigma_t)=\chi(p,W_t)\]
holds, where $S(\rho_t||p\otimes \sigma_t) $ is the relative entropy. 
\begin{lemma}[Donald's Inequality]\label{donald}
Consider any $t,t'\in T$. Then
\begin{eqnarray*} S(\rho_{t'}||p\otimes \sigma_t)\ge S(\rho_{t'}||p\otimes \sigma_{t'})  
\end{eqnarray*}
and equality holds iff $\sigma_{t'}=\sigma_t$.
\end{lemma}
\begin{proof}
The claimed inequality can be seen as a special instance of Donald's identity \cite{donald}. We give a short direct proof for reader's convenience. If $\mathrm{supp}(\rho_{t'})$ is not dominated by $ \mathrm{supp}(p\otimes \sigma_t)$ we have $S(\rho_{t'}||p\otimes \sigma_t)=+\infty$. But on the other hand $S(\rho_{t'}||p\otimes \sigma_{t'})=\chi(p,W_{t'})<+\infty$ for any $t'\in T$. Thus the claimed inequality is trivially fulfilled and is always strict in this case.\\
Assume now that $\mathrm{supp}(\rho_{t'})$ is dominated by $ \mathrm{supp}(p\otimes \sigma_t)$, then we obtain
\begin{eqnarray*}
 S(\rho_{t'}||p\otimes \sigma_t)&=&\textrm{tr}(\rho_{t'}\log\rho_{t'}-\rho_{t'}\log p\otimes \sigma_t) \\
&=& -S(\rho_{t'})-\textrm{tr}(\rho_{t'}\log p\otimes \sigma_t)\\
&=& -S(\rho_{t'})+S(p)-\textrm{tr}(\sigma_{t'}\log\sigma_t)\\
&=& -S(\rho_{t'})+S(p)-\textrm{tr}(\sigma_{t'}\log\sigma_t)\\
& & +\textrm{tr}(\sigma_{t'}\log \sigma_{t'})-\textrm{tr}(\sigma_{t'}\log \sigma_{t'})\\
&=& -S(\rho_{t'})+S(p)+S(\sigma_{t'})\\
& & +\textrm{tr}(\sigma_{t'}\log \sigma_{t'}-\sigma_{t'}\log \sigma_t)\\
&=& S(\rho_{t'}||p\otimes \sigma_{t'})+S(\sigma_{t'}||\sigma_{t})\\
&\ge& S(\rho_{t'}||p\otimes \sigma_{t'}),
\end{eqnarray*}
where we used the fact that $S(\sigma_{t'}||\sigma_{t})\ge 0 $ in the last line. We are done now since $S(\sigma_{t'}||\sigma_{t})=0 $ iff $\sigma_{t'}=\sigma_t$.
\end{proof}
\begin{remark}
A glance at the proof of Lemma \ref{donald} shows that the following stronger conclusion holds\footnote{We would like to thank the Associate Editor for pointing out this improvement of Lemma \ref{donald}}. For any $t\in T$ and \emph{any} state $\sigma\in \mathcal{S}(\hr)$
\[S(\rho_{t'}||p\otimes \sigma)\ge S(\rho_{t'}||p\otimes \sigma_{t'})  \]
with equality iff $\sigma=\sigma_{t'}$.
\end{remark}

For given $p\in\mathcal{P}(A)$ and $t\in T$ we set
\[S(\Omega_{p}||p\otimes \sigma_t):=\inf_{r\in T}S(\rho_r||p\otimes \sigma_t).   \]
\begin{lemma}\label{equality-2}
For each  $p\in\mathcal{P}(A)$ we have
\[ \inf_{t'\in T}S(\Omega_{p}||p\otimes \sigma_{t'})=\inf_{t'\in T}S(\rho_{t'}||p\otimes \sigma_{t'}).  \]
\end{lemma}
\begin{proof}
It is clear that $ \inf_{t'\in T}S(\Omega_{p}||p\otimes \sigma_{t'})\le \inf_{t'\in T}S(\rho_{t'}||p\otimes \sigma_{t'})$
holds. For the reverse inequality we choose an arbitrary $\eps >0$ and a $t(\eps)\in T$ with
\begin{equation}\label{trivial-1}
  S(\Omega_{p}||p\otimes \sigma_{t(\eps)})\le \inf_{t'\in T}S(\Omega_{p}||p\otimes \sigma_{t'})+\frac{\eps}{2},
\end{equation}
and a $s(\eps)\in V$ such that
\begin{eqnarray}\label{trivial-2}
  S(\rho_{s(\eps)}||p\otimes \sigma_{t(\eps)})&\le & S(\Omega_{p}||p\otimes \sigma_{t(\eps)})+\frac{\eps}{2}\nonumber\\
&\le& \inf_{t'\in T}S(\Omega_{p}||p\otimes \sigma_{t'})\nonumber\\
&&+\eps
\end{eqnarray}
where the last line follows from (\ref{trivial-1}). Donald's inequality, Lemma \ref{donald}, shows that $S(\rho_{s(\eps)}||p\otimes \sigma_{s(\eps)}) \le S(\rho_{s(\eps)}||p\otimes \sigma_{t(\eps)} )$,
and consequently by (\ref{trivial-2}) that
\[ \inf_{t'\in T}S(\rho_{t'}||p\otimes \sigma_{t'})\le \inf_{t'\in T}S(\Omega_{p}||p\otimes \sigma_{t'})+\eps  \]
holds for every $\eps>0$. This shows our claim.
\end{proof}
\subsection{The Direct Part of the Coding Theorem}
The crucial point in our code construction for the compound cq-channels will be following one-shot version of the coding theorem which is based on (and is an easy consequence of) the ideas developed by Hayashi and Nagaoka in \cite{hayashi-nagaoka}. In order to formulate the result properly we need some notation. Let $W:K\to\mathcal{S}(\mathcal{K})$ be any cq-channel with finite input alphabet $K$ and finite-dimensional output Hilbert space $\mathcal{K}$. Let $D_k:=W(k)$ for all $k\in K$. For any $w\in \mathcal{P}(K)$ we consider the states 
\[ \rho:=\sum_{k\in K}w(k)|k\rangle\langle k|\otimes D_k, \]
and $w\otimes \sigma $ with
\[\sigma=\sum_{k\in K}w(k)D_{k}  \]
acting on the Hilbert space $\cc^{|K|}\otimes \mathcal{K}$. Let $\mathcal{B}_{\textrm{diag}}$ denote  the set of operators on $\cc^{|K|}$ that are diagonal with respect to the orthonormal basis $\{|k\rangle  \}_{k\in K}$.
\begin{theorem}[Hayashi \& Nagaoka \cite{hayashi-nagaoka}]\label{one-shot-coding-theorem}
Given any cq-channel $W:K\to\mathcal{S}(\mathcal{K})$ and $w\in\mathcal{P}(K)$ with finite set $K$ and finite-dimensional Hilbert space $\mathcal{K}$. Let $P\in \mathcal{B}_{\textrm{diag}}\otimes \mathcal{B}(\mathcal{K})$ be a projection with
\begin{enumerate}
\item $\textrm{tr}(\rho P)\ge 1-\lambda$ with some $\lambda>0$ and
\item $\textrm{tr}((w\otimes \sigma)P)\le 2^{-\mu}$ for some $\mu>0$.
\end{enumerate}
Then for each $0<\gamma<\mu$ we can find $k_{1},\ldots ,k_{[2^{\mu-\gamma}]}\in K$ and $b_1,\ldots ,b_{[2^{\mu-\gamma}] }\in \mathcal{B}(\mathcal{K})$ with $b_i\ge 0$ and $\sum_{i=1}^{[2^{\mu-\gamma}] }b_i\le \idn_{\mathcal{K}} $ such that
\[\frac{1}{[2^{\mu-\gamma}] }\sum_{i=1}^{[2^{\mu-\gamma}] }(1-\textrm{tr}(D_{k_i}b_i ))\le 2\cdot\lambda +4\cdot 2^{-\gamma}. \]
\end{theorem}
\begin{proof}
All arguments needed in the proof of this theorem are contained explicitly or implicitly in \cite{hayashi-nagaoka}. We provide the proof in Appendix \ref{proof-one-shot-coding-theorem} for completeness and in order to make the presentation more self-contained.
\end{proof}
As in the classical approaches to the direct part of the coding theorem we need a discrete approximation of our compound cq-channel. A partition $\Pi$ of $\mathcal{S}(\hr)$ is a family $\{\pi_1,\ldots ,\pi_y  \}$ of subsets of $\mathcal{S}(\hr)$ such that $\pi_i\cap \pi_j=\emptyset $ for $i\neq j $ and $\mathcal{S}(\hr)=\bigcup_{i=1}^{y}\pi_i$ hold. We say that the diameter of the partition $\Pi=\{\pi_1,\ldots ,\pi_y  \}$ of $\mathcal{S}(\hr)$ is at most $\kappa>0$ if
\[\sup_{\rho,\sigma\in\pi_i}||\rho-\sigma||_1\le\kappa \qquad \forall i=1,\ldots, y.  \]
We borrow from \cite{winter-diss} a basic partitioning result for $\mathcal{S}(\hr)$ which is proven by a packing argument in the $d^{2}$-dimensional cube.
\begin{theorem}[Winter, Lemma II.8 in \cite{winter-diss}]
For any $\kappa>0$ there is a partition $\Pi=\{\pi_i,\ldots ,\pi_y  \} $ of $\mathcal{S}(\hr)$ having diameter at most $\kappa$ with $y\le K\kappa^{-d^2}$, where the number $K>0$ depends only on the dimension $d$ of $\hr$.
\end{theorem}
Applying this result $|A|$-times outputs for each $\kappa>0$ a partition $\Pi$ of the set of cq-channels $CQ(A,\hr)$ with input alphabet $A$ and output Hilbert space $\hr$ with at most $K^{|A|}\cdot \kappa^{-|A|d^{2}}$ elements.
For $n\in \nn$ we choose $\kappa=\kappa_n:=\frac{1}{n^2}$ and a partition $\Pi_{\kappa_n}=\{\pi_{1,n},\ldots ,\pi_{y,n}  \}$ of $CQ(A ,\hr)$ with at most $K^{|A|}\cdot n^{|A|d^{2}}$ elements and diameter not exceeding $\kappa_n$. This $\Pi_{\kappa_n}$ produces a partition 
\[  \Pi'_{n}:=\{ \pi_{i,n}\cap T: i=1,\ldots,y, \pi_{i,n}\cap T\neq \emptyset \},  \]
of the given compound cq-channel T. From each $\pi_{i,n}\cap T\neq \emptyset$ we select one cq-channel $W_{t_i}$ and denote this finite set of channels by $T'_n$.\\
Let $U:A\to \mathcal{S}(\hr)$ denote the useless cq-channel $U(x):=(1/d)\cdot \idn_{\hr}$. We set $W'_t:=(1-\frac{1}{n^2})W_t+\frac{1}{n^2}U$ for all $t\in T'_n$. The resulting set of channels will be denoted by $T_n$. Written in terms of density operators this defining relation means that we consider
\begin{equation}\label{approx-density}
 D'_{t,x}:=(1-\frac{1}{n^2})D_{t,x}+\frac{1}{n^2 d}\idn_{\hr},  
\end{equation}
for all $t\in T'_n$ and all $x\in A$.
\begin{lemma}\label{T-n-vs-T}
Let $T$ be any compound cq-channel and choose $n\in \nn$. Then the associated compound cq-channel $T_n$ has the following properties:
\begin{enumerate}
\item $|T_n|\le K^{|A|} \cdot n^{|A|d^{2}}$.
\item For each $t\in T$ we can find at least one $s\in T_n$ such that for all $x^n\in A^n$ 
\[||D_{t,x^n}-D'_{s,x^n}||_1\le \frac{4}{n},  \]
where $||\cdot||_{1}$ denotes the trace distance. The same statement holds if we reverse the roles of $t\in T$ and $s\in T_n$.
\item There is a constant $C=C(d)$ such that for each $p\in\mathcal{P}(A)$ and all $n\in \nn$
\[|\min_{s\in T_n}\chi(p,W'_s)-\inf_{t\in T}\chi(p,W_t)|\le C/n\]
holds.
\end{enumerate}
\end{lemma}
\begin{proof}
The first part of the lemma is clear by our construction of $T_n$.\\
The second assertion follows from the general fact that for states $\rho_1,\ldots ,\rho_n, \sigma_1,\ldots ,\sigma_n\in\mathcal{S}(\hr)$ the relation
\[||\rho_1\otimes \ldots \otimes \rho_n-\sigma_1\otimes \ldots \otimes \sigma_n||_1\le \sum_{i=1}^{n}||\rho_i-\sigma_i||_1 \]
holds and that for each $t\in T$ we can find $s'\in T'_n$ with $||D_{t,x}-D_{s',x}||_1\le 2/n^2$ for all $x\in A$ and to each $s'\in T'_n$ there is obviously $s\in T_n$ with $||D_{s',x}-D'_{s,x}   ||_1\le  2/n^2 $ for all $x\in A$ .\\
The last part of the lemma is easily deduced from the Fannes inequality \cite{fannes} which states that for any states $\rho,\sigma\in \mathcal{S}(\hr)$ with $||\rho-\sigma||_1\le \delta \le 1/e$ we have $|S(\rho)-S(\sigma) |\le \delta\log d -\delta\log\delta$.
 Indeed, for each $n\in \nn$ choose $s_n\in T_n$ with
\begin{equation}\label{T-n-vs-T-1}
  \chi(p,W'_{s_n})= \min_{s\in T_n}\chi(p,W'_t).
\end{equation}
Then observing that
\[\chi(p,W'_{s_n})=S(\sum_{x\in A}p(x)D'_{t_{n},x})-\sum_{x\in A}p(x)S(D'_{t_{n},x}),  \]
and that we can find $t\in T$ with $||D_{t,x}-D'_{s_n,x}||_1\le 4/n^2  $ for all $x\in A$ leads via Fannes inequality to
\begin{equation}\label{T-n-vs-T-2}
  |\chi(p,W'_{s_n})-\chi(p,W_{t})|\le 2(\frac{4}{n^2}\log d -\frac{4}{n^2}\log\frac{4}{n^2} ), 
\end{equation}
provided that $n\ge \sqrt{\frac{e}{4}}$. (\ref{T-n-vs-T-1}) and (\ref{T-n-vs-T-2}) show that
\begin{eqnarray*}
  \inf_{t\in T}\chi(p,W_t)&\le& \min_{s\in T_n}\chi(p,W'_s)\\
&& +2(\frac{2}{n^2}\log d -\frac{2}{n^2}\log\frac{2}{n^2} )\\
&=& \min_{s\in T_n}\chi(p,W'_s)+ O(n^{-1}).
 \end{eqnarray*}
A similar argument shows the reverse inequality and we are done.
\end{proof}
\begin{remark}
At this point we pause for a moment to indicate why our discretization Lemma \ref{T-n-vs-T} does not suffice to reduce the capacity problem for arbitrary sets of channels to the finite case solved by Datta and Dorlas \cite{datta-dorlas}. Let us assume that we want to construct codes for the channel $T_n$ of block length $n$
The proof strategy in \cite{datta-dorlas}, translated into the setting of our Lemma \ref{T-n-vs-T} would consist of a combination of a measurement that detects the branch from $T_n$ combined with reliable codes for individual channels from $T_n$. In order to detect which channel is in use during the transmission Datta and Dorlas construct a sequence $x^{mL_n}\in A^{mL_n}$, $L_n:=\binom{|T_n|}{2}$, and a PVM in $\{p_t^{mL_n}  \}_{t\in T_n}$ in $\mathcal{B}(\hr^{\otimes mL_n})$ with
\begin{equation}\label{datta-dorlas-exclusion}
\textrm{tr}(p_t^{mL_n}W_t^{m L_n}(x^{mL_n})  )\ge (1- |T_n|f^{m})^{|T_n|-1},  
\end{equation}
where $f\in (0,1)$. It is easily seen using standard volumetric arguments with respect to the Hausdorff measure on the set of cq-channels that for open sets $T$ (w.r.t. the relative topology) of channels $|T_n|\ge \textrm{poly}(n)$ with degree strictly larger than $1$. Hence, $L_n=\textrm{poly}(n)$. And since the rightmost quantity in (\ref{datta-dorlas-exclusion})
has to approach $1$ we have to choose $m=m(n)$ as an increasing sequence depending on $n$. Thus for large $n$ $m_n L_n=m_n\textrm{poly}(n)\ge n$ and no more block length is left for coding.  
\end{remark}
In the course of the proof of Theorem \ref{comp-direct-part} we will need two probabilistic inequalities which go back to the work of Blackwell, Breiman, and Thomasian \cite{blackwell} and Hoeffding \cite{hoeffding}.
Let $\{V_{t} \}_{t\in T}$ be a finite set of stochastic matrices $V_{t}:X\to J$ with finite sets $X$ and $J$. For $r\in \mathcal{P}(X)$ we set
\[p_t(x,j):=r(x)V_t(j|x) \qquad (x\in X,j\in J),   \]
and
\[q_t(j):=\sum_{x\in X}r(x)V_{t}(j|x).  \]
Moreover, for each $a\in \nn$ we define the averaged channel $V^a:X^a\to J^a$ by 
\[V^a(j^a|x^a ):=\frac{1}{|T|}\sum_{t\in T}V_t^a(j^a|x^a),\] 
the joint input-output distribution
\[p'^a(x^a,j^a):=r^{\otimes a}(x^a)V^a(j^a|x^a), \]
and 
\[q^a:=\frac{1}{|T|}\sum_{t\in T}q_{t}^{\otimes a}.\] 
For each $t\in T$ and $a\in \nn$ let
\begin{equation}\label{ind-info-t}
 i_{t}^{a}(x^a,j^a):=\frac{1}{a}\log \frac{V_{t}^{a}(j^a|x^a)}{q_{t}^{\otimes a}(j^a)}, 
\end{equation}
and
\begin{equation}\label{ind-info}
i^{a}(x^a,j^a):=\frac{1}{a}\log \frac{V^{a}(j^a|x^a)}{q^{a}(j^a)},  
\end{equation}
where $x^a\in X^a$ and $j^a\in J^a$.
\begin{theorem}[Blackwell, Breiman, Thomasian \cite{blackwell}]\label{finite-feinstein}
With the notation introduced in preceding paragraph we have for all $\alpha,\beta\in\rr$
\[\mathbb{P}(i^a\le \alpha  )\le \frac{1}{|T|}\sum_{t\in T}\mathbb{P}_t (i_t^a\le \alpha +\beta  )+|T|2^{-a\beta}.\]
\end{theorem}
Our proof of Theorem \ref{comp-direct-part} will also require Hoeffding's tail inequality:
\begin{theorem}[Hoeffding \cite{hoeffding}]\label{hoeffding}
Let $X_1,\ldots , X_a$ be independent real valued random variables such that each $X_i$ takes values in the interval $[u_i,o_i]$ with probability one, $i=1,\ldots ,a$. Then for any $\tau>0$ we have
\[ \mathbb{P}\left(\sum_{i=1}^{a}(X_i-\mathbb{E}(X_i))\ge a\tau\right)\le e^{-2\frac{a^2\tau^2}{\sum_{i=1}^{a}(o_i-u_i)^2}}  \]
and
\[ \mathbb{P}\left(\sum_{i=1}^{a}(X_i-\mathbb{E}(X_i))\le- a\tau\right)\le e^{-2\frac{a^2\tau^2}{\sum_{i=1}^{a}(o_i-u_i)^2}}  \]
\end{theorem}
With all these preliminary results we are able now to state and prove our main objective:
\begin{theorem}[Direct Part]\label{comp-direct-part}
Let $T$ be an arbitrary compound cq-channel. Then for each $\lambda\in (0,1)$ and any $\alpha>0$ we can find $(n,M_n,\lambda)_{\textrm{max}}$-codes with
\[\frac{1}{n}\log M_n\ge \max_{p\in \mathcal{P}(A)}\inf_{t\in T}\chi(p,W_t)-\alpha,  \]
for all $n\in\nn$ with $n\ge n_0(\alpha,\lambda)$. Consequently, for each $\lambda\in (0,1)$
\[ C(T,\lambda)\ge\max_{p\in \mathcal{P}(A)}\inf_{t\in T}\chi(p,W_t).  \]
\end{theorem}
\begin{proof}
Our strategy will be, roughly, to construct a ``good'' projection for the averaged channel $W^n=\frac{1}{|T_n|}\sum_{t\in T_n}{W'}_{t}^{n}$ via Theorem \ref{abelian-approximation}, Theorem \ref{finite-feinstein}, and Theorem \ref{hoeffding}. This means that for a suitably chosen input distribution $p\in \mathcal{P}(A)$, the associated state
\[\rho^{(n)}=\sum_{x^n\in A^n}p^{\otimes n}(x^n)|x^n\rangle\langle x^n|\otimes \sum_{t\in T_n}W_t^{n}(x^n)  \]
and the resulting product of the marginal states 
  \[p^{\otimes n}\otimes \sigma^{(n)}  \]
we will find a projection $P_n\in (\mathcal{A}_{\textrm{diag}}\otimes \mathcal{B}(\hr))^{\otimes n}$ with
\begin{enumerate}
\item $\textrm{tr}(\rho^{(n)}P_n)\approx 1$, and
\item $\textrm{tr}((p^{\otimes n}\otimes \sigma^{(n)}  ) P_n  )\lessapprox 2^{-n \inf_{t\in T}\chi(p,W_t)}  $.
\end{enumerate}
Then we will apply Theorem \ref{one-shot-coding-theorem} to obtain a good code for $W^n$. This code performs well for the compound channel $T_n$ since the error probability depends affinely on the channel. Finally, by Lemma \ref{T-n-vs-T} we see that the code obtained in this way is also reliable for the original channel $T$.\\
Let $p=\mathrm{argmax}_{p'\in \mathcal{P}(A)}(\inf_{t\in T}\chi(p',W_t)) $. We assume w.l.o.g. that $ \inf_{t\in T}\chi(p,W_t)>0 $, because otherwise the assertion of the theorem is trivially true.\\
Our goal is to construct $(n,M_n,\frac{\lambda}{2})_{\max}$-codes $\mathcal{C}_n$ for the approximating channel $T_n$ with
\[M_n\ge 2^{n(\inf_{t\in T}\chi(p,W_t)-\alpha )}  \]
for all sufficiently large $n\in \nn$. Then by Lemma \ref{T-n-vs-T} $\mathcal{C}_n$ is also an $(n, M_n, \frac{\lambda}{2}+\frac{4}{n})_{\max}$-code for the original channel $T$. Choosing $n$ large enough we can ensure that $\frac{4}{n}\le \frac{\lambda}{2}$ and our proof would be accomplished.\\ 
In what follows we use the abbreviations
\[\Omega_{p,n}:=\{\rho'_t: \rho'_t=\sum_{x\in A}p(x)|x\rangle\langle x|\otimes D'_{t,x}, t\in T_n \}  \]
and for $t\in T_n$ we write
\[\sigma'_t:=\sum_{x\in A}p(x)D'_{t,x},  \]
where $p\in \mathcal{P}(A)$ is arbitrary. Note that by (\ref{approx-density}) we have for each $t\in T_n$
\begin{equation}\label{comp-direct-3}
  \lambda_{\min}(p\otimes \sigma'_t)\ge p_{\min}\frac{1}{n^2d}.
\end{equation}
Moreover it is clear from the definition of $T_n$ that $\mathrm{supp}(\rho'_t)$ is dominated by $\mathrm{supp}(p\otimes \sigma'_s)$ for each $t,s\in T_n$ and $\mathrm{supp}(p\otimes \sigma'_s)=\mathrm{supp}(p)\otimes \idn_{\hr}$ for all $s\in T_n$. Now choose any $s\in T_n$. By the properties of the supports just mentioned we may assume w.l.o.g. that $p\otimes \sigma_s$ is invertible. Then for fixed $l\in\nn$ we can find $a,b\in \nn$ with $n=al+b$, $0\le b<l$, and obtain from Theorem \ref{abelian-approximation} a PVM $\mathcal{M}_l=\{P_{1,l},P_{2,l}\}$ with $P_{i,l}\in (\mathcal{A}_{\textrm{diag}}\otimes \mathcal{B}(\hr))^{\otimes l}$, $i=1,2$, with
\begin{eqnarray}\label{comp-direct-4}
  S_{\mathcal{M}_l}({\rho'}_{t}^{\otimes l}||(p\otimes \sigma'_s)^{\otimes l})&\ge& l(S(\Omega_{p,n}||p\otimes \sigma'_s)-\zeta_l(p\otimes \sigma'_s))\nonumber\\
&\ge&l(\min_{t\in T_n}\chi(p,W'_t)-\zeta_l(p\otimes \sigma'_s)),\nonumber\\
\end{eqnarray}
where we have used Lemma \ref{equality-2}.\\
Since $P_{i,l}\in (\mathcal{A}_{\textrm{diag}}\otimes \mathcal{B}(\hr))^{\otimes l}$ for $i=1,2$ we can find projections $\{ r_{i,x^l} \}_{x^l\in A^l}\subset \mathcal{B}(\hr)^{\otimes l}$, $i=1,2$, with
\[P_{i,l}=\sum_{x^l\in A^l}|x^l\rangle \langle x^l|\otimes r_{i,x^l} \qquad (i=1,2). \]
The relation 
\[(\idn_{\mathcal{A}_{\textrm{diag}}}\otimes \idn_{\hr})^{\otimes l}=P_{1,l}+P_{2,l}  \]
implies
\begin{equation}\label{comp-direct-5}
  \idn_{\hr}^{\otimes l}=r_{1,x^l}+r_{2,x^l} \qquad \forall x^l\in A^l.
\end{equation}
For each $x^l\in A^l$ let $\{e_{x^l,j}  \}_{j=1}^{\textrm{tr}(r_{1,x^l})}$ be an orthonormal basis of the range of $r_{1,x^l}$ and $\{ e_{x^l,j} \}_{j=\textrm{tr}(r_{1,x^l})+1}^{d^l}$ an orthonormal basis of the range of $r_{2,x^l}$. Then by (\ref{comp-direct-5}) the set $\{|x^l\rangle \otimes e_{x^l ,j}  \}_{x^l\in A^l,j=1}^{\ \qquad d^l}$ is an orthonormal basis of $(\cc^{|A|}\otimes \hr)^{\otimes l}$, and we have by definition
\[P_{1,l}=\sum_{x^l\in A^l}|x^l\rangle\langle x^l|\otimes \sum_{j=1}^{\textrm{tr}(r_{1,x^l})}|e_{x^l,j}\rangle\langle e_{x^l,j}|, \]
and similarly
\[P_{2,l}=\sum_{x^l\in A^l}|x^l\rangle\langle x^l|\otimes \sum_{j=\textrm{tr}(r_{1,x^l})+1}^{d^l}|e_{x^l,j}\rangle\langle e_{x^l,j}|,\]
i.e. the PVM $\mathcal{Q}_l(s):=\{|x^l\rangle\langle x^l|\otimes |e_{x^l,j}\rangle\langle e_{x^l,j}|  \}_{x^l\in A^l,j=1}^{\ \qquad d^l}$ consisting of one-dimensional projections is a refinement of the PVM $\mathcal{M}_l=\{P_{1,l},P_{2,l}\}$. Thus by the monotonicity of the relative entropy and (\ref{comp-direct-4}) we obtain
\begin{equation}\label{comp-direct-6}
  S_{\mathcal{Q}_l(s)}({\rho'}_{t}^{\otimes l}||(p\otimes \sigma'_s)^{\otimes l})\ge l(\min_{t\in T_n}\chi(p,W'_t)-\zeta_l(p\otimes \sigma'_s)),
\end{equation}
for all $t\in T_n$, and consequently
\begin{equation}\label{comp-direct-7}
 \min_{s\in T_n}\min_{t\in T_n}S_{\mathcal{Q}_l(s)}({\rho'}_{t}^{\otimes l}||(p\otimes\sigma'_s)^{\otimes l})\ge l(\min_{t\in T_n}\chi(p,W'_t)-\zeta_l(p) ), 
\end{equation}
where
\[ \zeta_l(p)=\max_{s\in T_n}\zeta_l(p\otimes \sigma'_s). \]
\emph{Claim:} For the choice $l=l_n=[\sqrt{n}]$ we have
\begin{equation}\label{zeta-to-zero}
\lim_{n\to\infty}\zeta_{l_n}(p)=0.  
\end{equation}
Recall from the proof of Theorem \ref{abelian-approximation} that
\[ \zeta_{l_{n}}(p\otimes\sigma'_s)= (1-\tau_{1,l_{n}})\tau_{2,l_{n}}(s)-\tau_{1,l_{n}}\log \lambda_{\min}(p\otimes \sigma'_s) +\frac{1}{l_{n}}, \]
where $\tau_{1,l}$ and $\tau_{2,l}=\tau_{2,l}(s)$ are defined in (\ref{abel-approx-2}) and (\ref{abel-approx-3}). Our remaining goal is to prove
\begin{equation}\label{comp-direct-8}
  \lim_{n\to\infty}\max_{s\in T_n}\tau_{2,l_{n}}(s)=0,
\end{equation}
and
\begin{equation}\label{comp-direct-9}
  \lim_{n\to\infty}\tau_{1,l_{n}}\max_{s\in T_n}(-\log \lambda_{\min}(p\otimes \sigma'_s))=0.
\end{equation}
In order to simplify the notation and streamline the subsequent arguments we introduce following terminology: Let $(a_n)_{n\in\nn}$ and $(b_n)_{n\in\nn}$ be two sequences of non-negative reals. We write $a_n\sim_{+} b_n$ if $\lim_{n\to\infty}\frac{a_n}{b_n}>0$.
The validity of the assertions (\ref{comp-direct-8}) and (\ref{comp-direct-9}) can be easily deduced from (\ref{comp-direct-3}) and the facts that $k_{l_{n}}\sim_{+}\frac{n^{1/2}}{\log n^{1/16}}$, $\delta_{l_n}\sim_{+}n^{-1/8}$, and $k_{l_n}\delta_{l_n}^{2}\sim_{+} \frac{n^{3/8}}{\log n^{1/16}} $. For example we have by (\ref{comp-direct-3})
\begin{eqnarray*}
0&\le& \tau_{1,l_{n}}\max_{s\in T_n}(-\log \lambda_{\min}(p\otimes \sigma'_s))\le -\tau_{1,l_n}\log \frac{p_{\min}}{n^2\cdot d}\\
&=& 2^{-k_{l_n}\delta_{l_n}^2(c-o(n^0)-\frac{1}{k_{l_n}\delta_{l_n}^2}\log \frac{n^2 d}{p_{\min}})},  
\end{eqnarray*}
which tends to $0$ as $n\to\infty$ since $k_{l_n}\delta_{l_n}^{2}\sim_{+} \frac{n^{3/8}}{\log n^{1/16}} $. Thus, (\ref{comp-direct-9}) is proven. In order to prove (\ref{comp-direct-8}) it suffices to show
that
\[\lim_{n\to\infty} \max_{s\in T_n}(-\delta_{l_n}\log \delta_{l_n}-\delta_{l_n}\log \lambda_{\min}(p\otimes \sigma'_s) )=0. \]
But this is clear from
\begin{eqnarray*} 
\begin{split}
\max_{s\in T_n}(-\delta_{l_n}\log \delta_{l_n}-\delta_{l_n}\log \lambda_{\min}(p\otimes \sigma'_s) )&\le& -\delta_{l_n}\log \delta_{l_n}\\
&&-\delta_{l_n}\log\frac{p_{\min}}{n^2 d}
\end{split} 
\end{eqnarray*}
and $\delta_{l_n}\sim_{+}n^{-1/8} $.\\
Choose $s^{*}\in T_n$ such that
\begin{equation}\label{comp-direct-10}
s^{*}= \mathrm{argmin}_{s\in T_n}(\min_{t\in T_n}S_{\mathcal{Q}_l(s)}({\rho'}_{t}^{\otimes l}||(p\otimes\sigma'_s)^{\otimes l})),
 \end{equation}
and consider the corresponding PVM $\mathcal{Q}_{l_{n}}(s^{*})=\{|x^{l_{n}}\rangle\langle x^{l_{n}}|\otimes |e_{x^{l_{n}},j}\rangle\langle e_{x^{l_{n}},j}|  \}_{x^{l_{n}}\in A^{l_{n}},j=1}^{\ \ \qquad \quad d^{l_{n}}}  $. For each $t\in T_n$ we define
\begin{eqnarray*} 
p_t(x^{l_{n}},j)&:=&\textrm{tr}({\rho'}_{t}^{\otimes l_{n}}|x^{l_{n}}\rangle\langle x^{l_n}|\otimes |e_{x^{l_n},j}\rangle\langle e_{x^{l_{n}},j}|  )\\
&=& p^{\otimes l_{n}}(x^{l_{n}})\textrm{tr}(D'_{t,x^{l_{n}}}|e_{x^{l_{n}},j}\rangle\langle e_{x^{l_{n}},j}|  ) \\ 
&=& p^{\otimes l_{n}}(x^{l_{n}})V_t(j|x^{l_{n}}),
\end{eqnarray*}
where for each $t\in T_n$ the stochastic matrix $V_t:A^{l_{n}}\to \{1,\ldots ,d^{l_{n}}\}$ is given by
\[V_t(j|x^{l_{n}}):= \textrm{tr}(D'_{t,x^{l_{n}}}|e_{x^{l_{n}},j}\rangle\langle e_{x^{l_{n}},j}|  ) \]
for $x^{l_{n}}\in A^{l_{n}}, j\in \{1,\ldots ,d^{l_{n}}\}$. By (\ref{comp-direct-10}), (\ref{comp-direct-7}), and (\ref{zeta-to-zero}) we get
\begin{equation}\label{comp-direct-11}
  \min_{t\in T_n}I(p^{\otimes l_{n}}, V_t)\ge l_{n}(\min_{t\in T_n}\chi(p,W'_t)-\zeta_{l_{n}}(p) ),
\end{equation}
with $\lim_{n\to\infty}\zeta_{l_{n}}(p)=0$. (\ref{comp-direct-11}) implies together with Lemma \ref{T-n-vs-T} that
\begin{equation}\label{comp-direct-12}
  \frac{1}{l_{n}}\min_{t\in T_n}I (p^{\otimes l_{n}},V_t)\ge \inf_{t\in T}\chi(p,W_t)-\frac{C}{n}-\zeta_{l_{n}}(p).
\end{equation}
This implies that we can find $n_1(\eps_1)$ such that 
\begin{equation}\label{comp-direct-13}
   \frac{1}{l_{n}}\min_{t\in T_n}I (p^{\otimes l_{n}},V_t)\ge \frac{1}{2}\inf_{t\in T}\chi(p,W_t)>0 
\end{equation}
for all $n\ge n_{1}(\eps_1)$. The last inequality in (\ref{comp-direct-13}) holds by our general assumption that $\inf_{t\in T}\chi(p,W_t)>0$. Choose any $n\ge n_1(\eps_1)$. 
Let
\[\Theta:=\left\{\theta\in\rr: 0<\theta <\frac{1}{6} \inf_{t\in T}\chi(p,W_t)\right\} \]
and
\begin{eqnarray}\label{comp-direct-14} 
I_n:&=&\min_{t\in T_n}I(p^{\otimes l_{n}},V_t)\nonumber\\
&=&\min_{s\in T_n}\min_{t\in T_n}D(p_t||r \otimes q_s),
\end{eqnarray}
where $r:=p^{\otimes l_{n}}$ and $q_t(j):=\sum_{x^{l_{n}}}r(x^{l_n})V_{t}(j|x^{l_{n}})$ for all $j\in \{1,\ldots,d^{l_{n}} \}$.
Moreover, in order to simplify our notation, we set $X:=A^{l_{n}}$ and $J:=\{1,\ldots ,d^{l_{n}}\}$ and suppress the $n$-dependence of $a$ and $l$ temporarily.\\
Recalling the definition of $i_t^a$ and $i^a$ from (\ref{ind-info-t}) and (\ref{ind-info}) we obtain from Theorem \ref{finite-feinstein} for $\alpha:=I_n-2l\theta$, $\beta:=l\theta$, $\theta \in\Theta$
\begin{equation}\label{comp-direct-15}
\mathbb{P}(i^a\le I_n-2l\theta )\le \frac{1}{|T_n|}\sum_{t\in T_n}\mathbb{P}_t (i_t^a\le I_n-l\theta  )+|T_n|2^{-al\theta}.
\end{equation}
Our construction of the compound cq-channel $T_n$ implies that for all $t\in T_n,x\in X,j\in J$
\[ V_{t}(j|x)\ge \frac{1}{(n^2d)^{l}}. \]
Consequently 
\[q_t(j)\ge  \frac{1}{(n^2d)^{l}}  \]
for all $j\in J$, and 
\begin{equation}\label{comp-direct-16}
  -l\log n^2d \le \log \frac{V_{t}(j|x)}{q_{t}(j)}\le l \log n^2d.
\end{equation}
Since $i_{t}^{a}$ is a sum of i.i.d. random variables each of which takes values in $[-l\log n^2d, l\log n^2d ]$ by (\ref{comp-direct-16}), we can apply Theorem \ref{hoeffding} and obtain
\begin{equation}\label{comp-direct-17}
\mathbb{P}_t(i_{t}^{a}\le I_n-l\theta)\le e^{-\frac{a l^2\theta^2}{4 l^2(\log n^2d)^2} } 
\end{equation}
for all $t\in T_n$ since $I_n\le \mathbb{E}_{t}(i_{t}^{a})$ for all $t\in T_n$. (\ref{comp-direct-17}) and (\ref{comp-direct-15}) show that
\begin{equation}\label{comp-direct-18}
  \mathbb{P}(i^a\le I_n-2l\theta )\le e^{-\frac{a\theta^2}{16 (\log nd)^2} } +|T_n|2^{-al\theta}.
\end{equation}
Thus the set $X_{a,\theta}\subset X^a\times J^a=A^{la}\times \{1,\ldots ,d^l  \}^a$ given by
\[X_{a,\theta}:=\{ (x^a,j^a): i^a(x^a,j^a)>I_n-l\theta \},  \]
is used to construct an orthogonal projection $P_{la,\theta}\in (\mathcal{A}_{\textrm{diag}}\otimes \mathcal{B}(\hr))^{\otimes la}$ defined by
\[ P_{la,\theta}:=\sum_{(x^a,j^a)\in X_{a,\theta}}|x^a\rangle\langle x^a|\otimes |e_{x^a,j^a}\rangle\langle e_{x^a,j^a}|,  \]
where we identify each $x^a\in X^a$ with a sequence in $A^{la}$. Moreover
\[ e_{x^a,j^a}:=e_{x_1,j_1}\otimes \ldots \otimes e_{x_a,j_a}.  \]
By the definition of set $X_{a,\theta}$ the relations
\begin{equation}\label{comp-direct-19}
  p'^a(X_{a,\theta})\ge 1- e^{-\frac{a\theta^2}{16 (\log nd)^2} } -|T_n|2^{-al\theta},
\end{equation}
and
\begin{equation}\label{comp-direct-20}
  (r^{\otimes a}\otimes q^a)(X_{a,\theta})\le 2^{-a(I_n-2l\theta)}
\end{equation}
hold. (\ref{comp-direct-19}) and (\ref{comp-direct-20}) imply by definition of the projection $P_{la,\theta}\in (\mathcal{A}_{\textrm{diag}}\otimes \mathcal{B}(\hr))^{\otimes la} $ that
\begin{equation}\label{comp-direct-21}
  \textrm{tr}(\rho^{(la)}P_{la,\theta} )\ge 1- e^{-\frac{a\theta^2}{16 (\log nd)^2} } -|T_n|2^{-al\theta},
\end{equation}
and
\begin{equation}\label{comp-direct-22}
  \textrm{tr}((p^{\otimes la}\otimes \sigma^{(la)} )P_{la,\theta})\le 2^{-a(I_n-2l\theta)},
\end{equation}
where
\begin{eqnarray*}
\rho^{(la)}&:=&\frac{1}{|T_n|}\sum_{t\in T_n}{\rho'}_{t}^{\otimes la}\\
&=&\sum_{x^{al}\in A^{al}}p^{\otimes al}(x^{al})|x^{al}\rangle\langle x^{al}|\otimes \frac{1}{|T_n|}\sum_{t\in T_n}D'_{t,x^{al}} ,  
\end{eqnarray*}
and
\[\sigma^{(la)}:=\frac{1}{|T_n|}\sum_{t\in T_n}{\sigma'}_{t}^{\otimes la}.  \]
Since $n=al+b$, $0\le b<l$, we can define a projection $P_{n,\theta}\in (\mathcal{A}_{\textrm{diag}}\otimes \mathcal{B}(\hr))^{\otimes n} $ by
\[P_{n,\theta}:=P_{la,\theta}\otimes (\idn_{\mathcal{A}_{\textrm{diag}}}\otimes \idn_{\hr})^{\otimes (n-la-1)},  \]
(\ref{comp-direct-21}), (\ref{comp-direct-22}) yield then 
\begin{equation}\label{comp-direct-23}
  \textrm{tr}({\rho}^{(n)}P_{n,\theta})\ge 1- e^{-\frac{a_n\theta^2}{16 (\log nd)^2} } -|T_n|2^{-a_nl_n\theta},
\end{equation}
and 
\begin{eqnarray}\label{comp-direct-24}
 \textrm{tr}((p^{\otimes n}\otimes \sigma^{(n)})P_{n,\theta})&\le& 2^{-a_n(I_n-2l_n\theta)}\nonumber\\
&\le & 2^{-a_nl_n(\inf_{t\in T}\chi(p,W_t)-\eps_n-2\theta)} \nonumber \\
\end{eqnarray}
by (\ref{comp-direct-12}) where $\eps_n:= \frac{C}{n}+\zeta_{l_n}(p)$. Thus for $n\ge n_2(\theta)$ we conclude from (\ref{comp-direct-24}), the fact that $\lim_{n\to\infty}\eps_n=0$, and $0\le b_n\le [n^{1/2}]$ that
\begin{equation}\label{comp-direct-25}
 \textrm{tr}((p^{\otimes n}\otimes \sigma^{(n)})P_{n,\theta})\le  2^{-n(\inf_{t\in T}\chi(p,W_t)-3\theta)}.
\end{equation}
Since the states $\rho^{(n)}\in (\mathcal{A}_{\textrm{diag}}\otimes \mathcal{B}(\hr))^{\otimes n}$ and $\sigma^{(n)}\in \mathcal{B}(\hr)^{\otimes n}$ correspond to the averaged cq-channel $W^{n}=\frac{1}{|T_n|}\sum_{t\in T_n}{W'}_t^n$ we can apply Theorem \ref{one-shot-coding-theorem} with 
\[\lambda=\lambda_n:= e^{-\frac{a_n\theta^2}{16 (\log nd)^2} } +|T_n|2^{-a_nl_n\theta},  \]
\[\mu=\mu_n:= n(\inf_{t\in T}\chi(p,W_t)-3\theta),  \]
\[ \gamma=\gamma_n=n\theta \]
and end up with a $(n,M'_n=[2^{n(\inf_{t\in T}\chi(p,W_t)-4\theta) }],\lambda'_n)_{\textrm{av}}$-code for the channel $W^{n}=\frac{1}{|T_n|}\sum_{t\in T_n}{W'}_t^n $ where
\[\lambda'_n=2\lambda_n+4\cdot 2^{-n\theta}.  \]
By standard arguments we can select a sub-code for $W^n$ with $M_n\ge (1/2)\cdot M'_n $ and maximum error probability $\tilde{\lambda}_n\le 2\lambda'_n$. We denote this $(n,M_n,\tilde{\lambda_n})_{\max}$-code by $\mathcal{C}_n$. But since
\[ W^{n}=\frac{1}{|T_n|}\sum_{t\in T_n}{W'}_t^n, \]
it is clear that $\mathcal{C}_n$ is a $(n,M_n, |T_n|\tilde{\lambda}_n)_{\max}$-code for the compound channel $T_n$. We know from our Lemma \ref{T-n-vs-T} that $|T_n|\le K^{|A|} n^{|A|d ^{2}}$. Thus since $l_n=[\sqrt{n}]$ and $a_n=\frac{n-b_n}{l_n}$ we see that
\[ \lim_{n\to\infty}|T_n|\tilde{\lambda}_n=0 \]
and we are done since $M_n\ge (1/2)[2^{n(\inf_{t\in T}\chi(p,W_t)-4\theta) }]\ge[2^{n(\inf_{t\in T}\chi(p,W_t)-5\theta) }]  $ for all sufficiently large $n\in \nn$. 
\end{proof}
\begin{remark}
Note that the error probability of the codes constructed in the proof of Theorem \ref{comp-direct-part} behaves like $1/n$ asymptotically. This is caused by our choice of $\tau_n$ as $\tau_n=1/n^{2}$. So we can achieve a faster decay of the decoding errors by using better sequences $\tau_n$. For example, if we choose $\tau_n=2^{-n^{1/16}}$ and replace $D'_{t,x}$ in (\ref{approx-density}) by
\[D'_{t,x}:=(1-\tau_n)D_{t,x}+\frac{\tau_n}{d}\idn_{\hr} \]
for all $x\in A$ and $t\in T'_n$ we obtain, as a careful inspection and a painless modification of the arguments applied so far show, for each sufficiently small $\theta >0 $ $(n,M_n,\lambda_n)_{\max}$-codes for the compound cq-channel $T$ with
\[M_n\ge [2^{n(\max_{p\in\mathcal{P}(A)}\inf_{t\in T}\chi(p,W_t)-5\theta )}]  \]
and
\[\lambda_n\le 2^{-c(\theta)n^{1/16}},  \]
for an appropriate positive constant $c(\theta)$.
\end{remark}
\subsection{The Strong Converse}
For the proof of the strong converse we simply follow Wolfowitz' strategy in \cite{wolfowitz-compound, wolfowitz}. To this end we use Winter's result from \cite{winter} which is the core of the strong converse for the single memoryless cq-channel:
\begin{theorem}[Winter \cite{winter}]\label{winter-theorem}
For $\lambda\in (0,1)$ there exists a constant $K'(\lambda,\dim \hr,|A|)$ such that for every memoryless cq-channel $\{ W^n \}_{n\in\nn}$ with finite input alphabet $A$ and finite-dimensional output Hilbert space $\hr$ and every $(n,M_n,\lambda)_{\max}$-code with the code words of the same type $p\in \mathcal{P}(A)$ the inequality
\[  M_n\le 2^{n(\chi(p,W)+K'(\lambda,\dim \hr,|A|)\frac{1}{\sqrt{n}}  )} \]
holds.
\end{theorem}
The proof of this theorem is implicit in the proof of Theorem 13 in \cite{winter}. 
\begin{theorem}[Strong Converse]\label{comp-strong-converse}
Let $\lambda\in (0,1)$. Then there is a constant $K=K(\lambda,\dim \hr,|A|) $ such that for any compound cq-channel $\{W_t^n  \}_{t\in T,n\in\nn}$ and any $(n,M_n,\lambda)_{\max}$-code $\mathcal{C}_n$
\[ \frac{1}{n}\log M_n\le \max_{p\in \mathcal{P}(A)}\inf_{t\in T}\chi(p,W_t)+K\frac{1}{\sqrt{n}} \]
holds.
\end{theorem}
\begin{proof}
Wolfowitz' proof of the strong converse \cite{wolfowitz-compound, wolfowitz} for the classical compound channel extends \emph{mutatis mutandis} to the cq-case once we have Theorem \ref{winter-theorem}.\\ 
We fix $n\in \nn$ and consider any $(n,M_n,\lambda)_{\max}$-code $\mathcal{C}_n=(u_i,b_i)_{i=1}^{M_n}$. Each code word $u_i\in A^n$ induces a type (empirical distribution) $p_{u_i}$ on $\mathcal{P}(A)$ and according to the standard type counting lemma (cf. \cite{csiszar}) there are at most $(n+1)^{|A|}$ different types. We divide our code $\mathcal{C}_n$ into sub-codes $\mathcal{C}_{n,j}=(u'_{k},b'_k)_{k=1}^{M_{n,j}}$ such that the code words of each $\mathcal{C}_{n,j}$ belong to the same type class, i.e. induce the same type. It is clear that the maximum error probabilities of these sub-codes are bounded from above by $\lambda$ for all $t\in T$. Since we have a uniform bound on error probabilities on each channel in the class $T$ we may apply Winter's, Theorem \ref{winter-theorem}, and obtain
\begin{equation}\label{str-conv-comp-1}
 M_{j}\le 2^{n(\chi(p_j,W_t)+K'(\lambda,\dim \hr, |A|)\frac{1}{\sqrt{n}} )}\quad \forall t\in T,
\end{equation}
where $p_j$ denotes the type of the code words belonging to the sub-code $\mathcal{C}_{n,j}$. Since the left hand side of (\ref{str-conv-comp-1}) does not depend on $t$ we may conclude that
\begin{eqnarray}\label{str-conv-comp-2}
 M_j&\le& 2^{n(\inf_{t\in T}\chi(p_j,W_t)+K'(\lambda,\dim \hr, |A|)\frac{1}{\sqrt{n}}  ) }\nonumber\\
&\le & 2^{n(\max_{p\in \mathcal{P}(A)}\inf_{t\in T}\chi(p,W_t)+K'(\lambda,\dim \hr, |A|)\frac{1}{\sqrt{n}}  ) }\nonumber\\
\end{eqnarray}
holds. Then, recalling that there are at most $(n+1)^{|A|}$ sub-codes and using (\ref{str-conv-comp-2}) we arrive at
\begin{eqnarray*}
\begin{split}
M_n \le  (n+1)^{|A|} 2^{n( \max_{p\in \mathcal{P}(A)}\inf_{t\in T}\chi(p,W_t)  +K'\frac{1}{\sqrt{n}}   ) } \\
\le  2^{n (\max_{p\in \mathcal{P}(A)}\inf_{t\in T}\chi(p,W_t) + K\frac{1}{\sqrt{n}}  )},
\end{split}
\end{eqnarray*}
with a suitable constant $K=K(\lambda,\dim \hr,|A|) $.
\end{proof}

\section{Averaged Channels}\label{av-channel}
In this section we extend the results of Datta and Dorlas \cite{datta-dorlas} to arbitrary averaged channels whose branches are memoryless cq-channels. \\
Let $(T,\Sigma, \mu)$ be a probability space, i.e. $T$ is a set, $\Sigma$ is a $\sigma$-algebra, and $\mu$ is a probability measure on $\Sigma$. Moreover we consider a memoryless compound cq-channel $\{W_t^{n}  \}_{t\in T,n\in \nn}$ with finite input alphabet $A$ and finite-dimensional output Hilbert space $\hr$. We assume that the branches $W_t$, $t\in T$, depend measurably on $t\in T$, i.e. we assume that for each fixed $x\in A$  the maps $T\ni t\mapsto D_{t,x}\in\mathcal{S}(\hr)$ are measurable. We assume here that $\mathcal{S}(\hr)$ is endowed with its natural Borel $\sigma$-algebra.\\
The averaged channel $W=\{W^n\}_{n\in\nn} $ is defined by the following prescription: For any $n\in \nn$ we have a map $W^n:A^n\ni x^n\mapsto D_{x^n}\in \mathcal{S}(\hr^{\otimes n})$ where $D_{x^n}$ is the density operator uniquely determined by the requirement that for all $b\in \mathcal{B}(\hr^{\otimes n})$ the relation
\[ \textrm{tr}(D_{x^n}b)=\int\textrm{tr}( D_{t,x^n}b)\mu(dt)  \]
holds\footnote{Note that $\textrm{tr}( D_{t,x^n}b) $ depends measurably on $t$ since tensor and ordinary products of operators are continuous and hence measurable operations.}.\\
A code $\mathcal{C}_n=(x^{n}(i),b_i)_{i=1}^{M_n}$ for the averaged channel $\{ W^n \}_{n\in\nn}$ consists as before of codewords $x^{n}(i)\in A^n$ and decoding operators $b_i\in \mathcal{B}(\hr)^{\otimes n}$, $b_i\ge 0$, $\sum_{i=1}^{M_n}b_i\le \idn_{\hr}^{\otimes n}$. The integer $M_n$ is the size of the code. Achievable rates and the capacity $C(W)$ are defined in a similar fashion as for memoryless cq-channels.\\
We will show in the following two subsections that, in analogy to the classical case \cite{ahlswede}, the weak capacity of $W$ is given by
\begin{equation}\label{averaged-capacity}
C(W)=\sup_{p\in\mathcal{P}(A)}\mathrm{ess-}\inf_{t\in T}\chi(p,W_t),  
\end{equation}
where $\mathrm{ess-}\inf$ denotes the essential infimum\footnote{The essential infimum of a measurable function $f:T\to\rr$ on the probability space $(T,\Sigma,\mu)$ is defined by $\mathrm{ess-}\inf_{t\in T}f:=\sup\{c\in \rr: \mu(\{t\in T: f(t)<c  \})=0\}$.}. Clearly, we cannot expect the strong converse to hold because of Ahlswede's \cite{ahlswede} counter examples in the classical setting.
\subsection{The direct part of the Coding Theorem}
We will need some simple properties of the essential infimum in the proof of the direct part of the coding theorem for the averaged channel $W$. We start with a simple general property of the essential infimum:
\begin{lemma}\label{ess-inf-prop-1}
Let $(T,\Sigma,\mu)$ be a probability space and $f:T\to \rr$ any measurable function. Let $ a:=\mathrm{ess-}\inf_{t\in T}f$.
Then the set $A:=\{ t\in T: f(t)\ge a  \}$ satisfies
\[ \mu(A)=1. \]
\end{lemma}
\begin{proof}
The assertion of the lemma follows easily from the definition of the essential infimum.
\end{proof}
Our proof of the direct part of the coding theorem will be based on a reduction to the case of compound cq-channels. Therefore we have to give another characterization of
\[\sup_{p\in\mathcal{P}(A)}\mathrm{ess-}\inf_{t\in T}\chi(p,W_t)  \]
in terms of the optimization processes appearing in the capacity formula for the compound cq-channels. To this end we define for any $p\in \mathcal{P}(A)$ 
\[a(p):=\mathrm{ess-}\inf_{t\in T}\chi(p,W_t), \]
and 
\[T_p:=\{t\in T: \chi(p,W_t)\ge a(p)  \}.\]
\begin{lemma}\label{ess-inf-prop-2}
Let $\{ W^n \}_{n\in \nn}$ be the averaged cq-channel defined by the probability space $(T,\Sigma,\mu)$ and the compound cq-channel $T$. Then
\[ \sup_{p\in\mathcal{P}(A)}\max_{q\in \mathcal{P}(A)}\inf_{t\in T_p}\chi(q,W_t)=\sup_{q\in \mathcal{P}(A)}\mathrm{ess-}\inf_{t\in T}\chi(q,W_t).  \]
\end{lemma}
\begin{proof}
$ \mu(T_p)=1$ holds by Lemma \ref{ess-inf-prop-1}. For $p,q\in \mathcal{P}(A)$ and the corresponding sets $T_p,T_q\subseteq T$ we have
\begin{eqnarray}\label{ess-inf-2-1}
  \inf_{t\in T_p}\chi(q,W_t)&\le& \inf_{t\in T_p\cap T_q}\chi(q,W_t)\nonumber\\
&\le& \mathrm{ess-}\inf_{t\in T}\chi(q,W_t),
\end{eqnarray}
where the last inequality is justified by the observation that $\mu (T_p\cap T_q)=1$ and that $T_p\cap T_q\subseteq \{t\in T: \chi(q,W_t)\ge \inf_{t\in T_p\cap T_q}\chi(q,W_t)  \}$, i.e. $\mu (\{t\in T: \chi(q,W_t)< \inf_{t\in T_p\cap T_q}\chi(q,W_t)  \})=0$ and (\ref{ess-inf-2-1}) holds by definition of the essential infimum. (\ref{ess-inf-2-1}) implies that
\[ \max_{q\in \mathcal{P}(A)}\inf_{t\in T_p}\chi(q,W_t)\le \sup_{q\in \mathcal{P}(A) }\mathrm{ess-}\inf_{t\in T}\chi(q,W_t), \]
and consequently
\begin{equation}\label{ess-inf-2-2}
  \sup_{p\in\mathcal{P}(A)}\max_{q\in \mathcal{P}(A)}\inf_{t\in T_p}\chi(q,W_t)\le\sup_{q\in \mathcal{P}(A)}\mathrm{ess-}\inf_{t\in T}\chi(q,W_t).
\end{equation}
In order to show the reverse inequality we choose for any $\eps>0$ a $q_{\eps}\in \mathcal{P}(A)$ with
\begin{equation}\label{ess-inf-2-3}
 \sup_{q\in \mathcal{P}(A)}\mathrm{ess-}\inf_{t\in T}\chi(q,W_t)\le \mathrm{ess-}\inf_{t\in T}\chi(q_{\eps},W_t)+\eps. 
\end{equation}
By definition of the set $T_{q_{\eps}}$ as
\[T_{q_{\eps}}=\{t \in T: \chi(q_{\eps},W_t)\ge a(q_{\eps})  \},  \]
with $a(q_{\eps})=\mathrm{ess-}\inf_{t\in T}\chi(q_{\eps},W_t)$ we have
\begin{equation}\label{ess-inf-2-4}
 \mathrm{ess-}\inf_{t\in T}\chi(q_{\eps},W_t)\le \inf_{t\in T_{q_{\eps}}}\chi(q_{\eps},W_t). 
\end{equation}
The inequalities (\ref{ess-inf-2-3}) and (\ref{ess-inf-2-4}) show that
\[ \sup_{q\in \mathcal{P}(A)}\mathrm{ess-}\inf_{t\in T}\chi(q,W_t)\le \inf_{t\in T_{q_{\eps}}}\chi(q_{\eps},W_t)+\eps, \]
which in turn yields
\begin{eqnarray*}
\sup_{q\in \mathcal{P}(A)}\mathrm{ess-}\inf_{t\in T}\chi(q,W_t)&\le&\sup_{p\in\mathcal{P}(A)}\max_{q\in\mathcal{P}(A)} \inf_{t\in T_{p}}\chi(q,W_t)\\ 
&& + \eps.
 \end{eqnarray*}
Since $\eps>0$ can be made arbitrarily small and the left hand side of the last inequality does not depend on $\eps$ we finally obtain
\[\sup_{q\in \mathcal{P}(A)}\mathrm{ess-}\inf_{t\in T}\chi(q,W_t)\le\sup_{p\in\mathcal{P}(A)}\max_{q\in\mathcal{P}(A)} \inf_{t\in T_{p}}\chi(q,W_t),  \]
which concludes our proof.
\end{proof}
\begin{theorem}[Direct Part]\label{direct-averaged}
Let $W$ denote the averaged cq-channel. Then
\[C(W)\ge \sup_{p\in\mathcal{P}(A)}\mathrm{ess-}\inf_{t\in T}\chi(p,W_t)  \] 
\end{theorem}
\begin{proof}
 We assume that 
\[\sup_{p\in\mathcal{P}(A)}\mathrm{ess-}\inf_{t\in T}\chi(p,W_t)>0\] 
since otherwise the assertion of the theorem is trivially true. \\
By Lemma \ref{ess-inf-prop-2} it is enough to show that for each $p\in \mathcal{P}(A)$ with
\[ \max_{q\in \mathcal{P}(A)}\inf_{t\in T_p}\chi(q,W_t)>0 \]
the rate
\[ \max_{q\in \mathcal{P}(A)}\inf_{t\in T_p}\chi(q,W_t)-\eps \]
is achievable for each sufficiently small $\eps >0$. But this follows immediately if we apply our Theorem \ref{comp-direct-part} to the compound channel $T_p$ since any good code for the compound cq-channel $T_p$ has the same performance for the averaged channel $W^n$ due to the fact that $\mu(T_p)=1$.
\end{proof}
\subsection{The Weak Converse}
We start with a general property of the essential infimum which will help us to reduce the arguments in the proof of the weak converse to Fano's inequality and Holevo's bound via Markov's inequality.
\begin{lemma}\label{prop-ess-inf}
Consider a probability space $(T,\Sigma,\mu)$. Let $n\in\nn$ and $f,f_n:T\to\rr$ be measurable bounded functions with
\begin{equation}\label{prop-ess-inf-vor}
  \lim_{n\to\infty}f_n(t)=f(t)\qquad \forall t\in T.
\end{equation}
Let $(G_n)_{n\in\nn}$ be a sequence of measurable subsets of $T$ with
\[\lim_{n\to\infty}\mu(G_n)=1.\]
Then
\begin{equation}\label{prop-ess-inf-concl}
  \limsup_{n\to\infty}\inf_{t\in G_n}f_n(t)\le \mathrm{ess-}\inf_{t\in T} f
\end{equation}
holds.
\end{lemma}
\begin{proof}
The proof will be accomplished if we can show the following two inequalities:
\begin{equation}\label{ess-inf-1}
  \limsup_{n\to\infty}\inf_{t\in G_n}f_n(t)\le \limsup_{n\to\infty}\inf_{t\in G_n}f(t),
\end{equation}
and
\begin{equation}\label{ess-inf-2}
 \limsup_{n\to\infty}\inf_{t\in G_n}f(t)\le \mathrm{ess-}\inf_{t\in T}f. 
\end{equation}
\emph{Proof of (\ref{ess-inf-1}):} Set
\[b_n:=\inf_{t\in G_n}f(t)\  \textrm{ and }\ b'_{n}:=\inf_{t\in G_n}f_n(t) .\]
Then to any $\eps>0$ we can find a $t_{\eps}\in G_n$ with
\begin{equation}\label{ess-inf-3}
  f(t_{\eps})\le b_n+\eps,
\end{equation}
and, by (\ref{prop-ess-inf-vor}), there is $n(\eps)\in \nn$ such that for all $n\ge n(\eps)$ we have
\begin{equation}\label{ess-inf-4}
  f_n(t_{\eps})\le f(t_{\eps})+\eps.
\end{equation}
Then the definition of $b'_{n}$, (\ref{ess-inf-4}), and (\ref{ess-inf-3}) yield
\[ b'_{n}\le b_{n}+2\eps \]
for all $n\ge n(\eps)$. This implies
\[ \limsup_{n\to\infty}b'_{n}\le \limsup_{n\to\infty}b_n +2\eps, \]
and since $\eps>0$ is arbitrary we obtain (\ref{ess-inf-1}).\\
\emph{Proof of (\ref{ess-inf-2}):} As in the first part of the proof we use the abbreviation 
\[ b_n:=\inf_{t\in G_n}f(t), \]
and additionally we set 
\[b:=\limsup_{n\to\infty}b_n.\]
Then by the very basic properties of the upper limit we can select a subsequence $(n_i)_{i\in \nn}$ with
\begin{equation}\label{ess-inf-5}
\lim_{i\to\infty}b_{n_{i}}=b. 
\end{equation}
In order to keep the notation as simple as possible we will denote this induced sequence $(b_{n_{i}})_{i\in\nn}$ by $(b_n)_{n\in\nn}$, i.e. we simply rename the subsequence. For any fixed $n\in \nn$ we consider the sequence $(A_{n,k})_{k\in\nn}$ consisting of measurable subsets of $T$ defined by $A_{n,k}:=\bigcup_{i=1}^{k}G_{n+i}$. Note that for each $n\in \nn$ the sequence $(A_{n,k})_{k\in\nn} $ has the following properties which are easy to check:
\begin{enumerate}
\item $A_{n,1}\subset A_{n,2}\subset\ldots $,
\item $\lim_{k\to\infty}\mu(A_{n,k})=1$,
\item $a_{n,k}:=\inf_{t\in A_{n,k}}f(t)=\min \{b_{n+1},b_{n+2},\ldots ,b_{n+k}\}$, the sequence $(a_{n,k})_{k\in\nn}$ is non-increasing for any $n\in \nn$, and
\item for $A_n:=\bigcup_{k\in\nn}A_{n,k}$ and $a_n:=\inf_{t\in A_n}f(t)$ we have $\mu(A_n)=1$, $a_n\le \mathrm{ess-}\inf_{t\in T}f$, and $a_n=\lim_{k\to\infty}a_{n,k}$ for each $n\in \nn$.
\end{enumerate}
In spite of these properties it suffices to prove that for each $\eps>0$ there is $n(\eps)\in\nn$ such that
\begin{equation}\label{ess-inf-6}
  b-\eps \le a_{n(\eps),k}\le b+\eps \qquad \forall k\in\nn,
\end{equation}
holds. In fact, (\ref{ess-inf-6}) implies then that
\[b-\eps\le a_{n(\eps)}\le b+\eps,  \]
since $a_{n(\eps)}=\lim_{k\to\infty}a_{n(\eps),k}$ and by choosing an appropriate sequence $(\eps_j)_{j\in\nn}$ with $\eps_j\searrow 0$ we can conclude that
\[ b=\limsup_{j\to\infty}a_{n(\eps_{j})}.\]
But then $b\le \textrm{ess-}\inf_{t\in T}f$ by $a_{n(\eps_j)}\le \textrm{ess-}\inf_{t\in T}f $ for all $j\in \nn$.\\
Thus we only need to prove (\ref{ess-inf-6}) which follows from (\ref{ess-inf-5}) (with our convention to suppress the index $i$):
To any $\eps>0$ we can find by (\ref{ess-inf-5}) an $n(\eps)\in\nn$ such that for all $n\ge n(\eps)$ we have
\[b-\eps\le b_n\le b+\eps.\]
Then by property 3) above we obtain for each $k\in\nn$
\[b-\eps\le \min\{b_{n(\eps)+1},\ldots ,b_{n(\eps)+k}\}=a_{n(\eps),k}\le b+\eps,\]
which is the desired relation.  
\end{proof}
As a last preliminary result we need the generalization of Lemma 6 in \cite{blackwell}.
\begin{lemma}\label{blackwell-lemma}
Let $\{W^n\}_{n\in\nn}$ be a memoryless cq-channel with input alphabet $A$ and output Hilbert space $\hr$. Then for any $(n,M_n,\eps_n)_{\mathrm{av}}$-code $\mathcal{C}_{n}=(x^n(i),b_i)_{i=1}^{M_n}$ with distinct codewords we have
\[(1-\eps_n)\log M_n\le n\chi(p_{\ast},W)+1,\]
where $p_{\ast}=\frac{1}{M_n}\sum_{i=1}^{M_n}p_{x^{n}(i)}\in \mathcal{P}(A)$ with empirical distributions or types $p_{x^{n}(i)}\in\mathcal{P}(A)$ of the codewords $x^n(i)$ for $i=1,\ldots ,M_n$.
\end{lemma}
\begin{proof}
The proof is based upon similar arguments as that of corresponding Lemma 6 in \cite{blackwell}. The only additional argument we need is Holevo's bound. The details are as follows; We may assume w.l.o.g. that $\sum_{i=1}^{M_n}b_i=\idn^{\otimes n}$ and define corresponding classical channel by
\[K(j|i):=\textrm{tr}(D_{x^{n}(i)}b_j)\qquad i,j\in \{1,\ldots,M_n \}.  \]
Let $\nu\in\mathcal{P}(A^n)$ be given by $\nu(x^n)=\frac{1}{M_n}$ if $x^n$ is one of $x^{n}(i)$, $i=1,\ldots, M_n$, and $\nu(x^n)=0$ else. In what follows we consider the marginal distributions $\nu_{1},\ldots,\nu_n\in\mathcal{P}(A)$ induced by $\nu\in\mathcal{P}(A^n)$. It is obvious that
\begin{equation}\label{blackwell-lemma-1}
  p_{\ast}(a)=\frac{1}{n}\sum_{j=1}^{n}\nu_j(a)\qquad \forall a\in A
\end{equation}
holds. From Fano's inequality and Holevo's bound we obtain
\begin{equation}\label{blackwell-lemma-2}
  (1-\eps_n)\log M_n\le I(\nu,K)+1\le \chi(\nu, W^{n})+1,
\end{equation}
where $I(\nu,K) $ denotes the mutual information evaluated for the input distribution $\nu$ and the classical channel $K$. Using the super-additivity (cf. \cite{holevo}) and concavity (w.r.t. the input distribution) of the Holevo information we get 
\begin{equation}\label{blackwell-lemma-3}
  \chi(\nu, W^n)\le \sum_{j=1}^{n}\chi(\nu_j,W)\le n\chi(p_{\ast},W),
\end{equation}
where we have used (\ref{blackwell-lemma-1}) in the last inequality. Inserting (\ref{blackwell-lemma-3}) into (\ref{blackwell-lemma-2}) yields the claimed relation.
\end{proof}
The corresponding weak converse is the content of the next theorem.
\begin{theorem}[Weak Converse]\label{aver-weak-converse}
Let $W$ be the averaged channel defined by the probability space $(T,\Sigma,\mu)$ and the compound channel $T$. Then any sequence $(\mathcal{C}_n)_{n\in\nn}$ of $(n,M_n,\eps_n)_{\mathrm{av}/\mathrm{max}}$-codes with $\lim_{n\to\infty}\eps_n=0$ fulfills
\[ \limsup_{n\to\infty}\frac{1}{n}\log M_n\le \sup_{p\in\mathcal{P}(A)}\mathrm{ess-}\inf_{t\in T}\chi(p,W_t). \]
\end{theorem}
\begin{proof}
Let $(\mathcal{C}_n)_{n\in\nn}$ be any sequence of $(n,M_n,\eps_n)_{\textrm{av}}$-codes with $\lim_{n\to\infty}\eps_n=0$, i.e.
\[\int e_{\textrm{av}}(t,\mathcal{C}_n)\mu(dt)=\eps_n,  \]
where
\[e_{\textrm{av}}(t,\mathcal{C}_n)=\frac{1}{M_n}\sum_{i=1}^{M_n}(1-\textrm{tr}(D_{t,x^{n}(i)}b_i)).\]
Set
\begin{equation}\label{aver-weak-conv-0}
G_n:=\{t\in T:e_{\textrm{av}}(t,\mathcal{C}_n)\le \sqrt{\eps_n}\}.
\end{equation}
Then Markov's inequality yields
\begin{equation}\label{aver-weak-conv-1}
  \mu(G_n)\ge 1-\sqrt{\eps_n}.
\end{equation}
If we choose $n_1\in\nn$ such that $\sqrt{\eps_n}<\frac{1}{2}$ for all $n\ge n_1$ then all the code words are distinct and we can apply Lemma \ref{blackwell-lemma} to each $t\in G_n$ (cf. (\ref{aver-weak-conv-0})) leading to
\[(1-\sqrt{\eps_n})\log M_n\le n\chi(p_{\ast},W_t)+1,\]
which is equivalent to
\begin{equation}\label{aver-weak-conv-2}
  \frac{1}{n}\log M_n\le \frac{\chi(p_{\ast},W_t)+\frac{1}{n}}{1-\sqrt{\eps_n}},
\end{equation}
for all $t\in G_n$ and all $n\ge n_1$. Since (\ref{aver-weak-conv-2}) holds for all $t\in G_n$ we obtain
\begin{equation}\label{aver-weak-conv-3}
  \frac{1}{n}\log M_n\le \frac{\inf_{t\in G_n}\chi(p_{\ast},W_t)+\frac{1}{n}}{1-\sqrt{\eps_n}}. 
\end{equation}
Recall that $p_{\ast}$ depends on the block length $n$. Thus we are done if we can show that
\begin{equation}\label{aver-weak-conv-4}
 \limsup_{n\to\infty}\max_{p\in \mathcal{P}(A)}\inf_{t\in G_n}\chi(p,W_t)\le \sup_{p\in\mathcal{P}(A)}\mathrm{ess-}\inf_{t\in T}\chi(p,W_t) 
\end{equation}
holds.\\
For each $n\in\nn$ with $n\ge n_1$ we choose $p_n\in\mathcal{P}(A)$ with
\[\inf_{t\in G_n}\chi(p_n , W_t)= \max_{p\in \mathcal{P}(A)}\inf_{t\in G_n}\chi(p,W_t). \]
By passing to a subsequence if necessary we may assume that
\begin{equation}\label{aver-weak-conv-5}
  \lim_{n\to\infty}\inf_{t\in G_n}\chi(p_n , W_t)= \limsup_{n\to\infty}\max_{p\in \mathcal{P}(A)}\inf_{t\in G_n}\chi(p,W_t).
\end{equation}
By selecting a further subsequence we can even ensure that $\lim_{j\to\infty}p_{n_{j}}=:p'\in \mathcal{P}(A)$ due to the compactness of $\mathcal{P}(A)$. By (\ref{aver-weak-conv-5}) we have
\begin{equation}\label{aver-weak-conv-6}
\lim_{j\to\infty}\inf_{t\in G_{n_{j}}}\chi(p_{n_{j}},W_t)=   \limsup_{n\to\infty}\max_{p\in \mathcal{P}(A)}\inf_{t\in G_n}\chi(p,W_t).
\end{equation}
Now, since
\[ \lim_{j\to\infty}\chi(p_{n_{j}},W_t)=\chi(p',W_t) \]
for all $t\in T$ by the continuity of Holevo information, and since $\lim_{j\to \infty}\mu(G_{n_{j}})=1$ by (\ref{aver-weak-conv-1}), we see that the assumptions of Lemma \ref{prop-ess-inf} are fulfilled for the functions
\[ f_j(t):=\chi(p_{n_{j}},W_t)\textrm{ and } f(t):=\chi(p',W_t).\]
Thus Lemma \ref{prop-ess-inf} and (\ref{aver-weak-conv-6}) show that 
\begin{eqnarray*}
 \limsup_{n\to\infty}\max_{p\in \mathcal{P}(A)}\inf_{t\in G_n}\chi(p,W_t)&\le &\mathrm{ess-}\inf_{t\in T}\chi(p',W_t)\\ 
&\le & \sup_{p\in\mathcal{P}(A)}\mathrm{ess-}\inf_{t\in T}\chi(p,W_t).
\end{eqnarray*}
This is exactly (\ref{aver-weak-conv-4}) and we are done.
\end{proof}
\section{Conclusion}
In this paper we have shown the existence of universally ``good'' classical-quantum codes for two particularly interesting cq-channel models with limited channel knowledge. We determined the optimal transmission rates for the classes of compound and averaged cq-channels. For the first model we could prove the strong converse for the maximum error criterion whereas for the latter only a weak converse is established.\\
The coding theorems for compound and averaged cq-channels imply in an obvious way the corrsponding capacity formulas for the classical product state capacities of compound and averaged quantum channels (cf. the arguments in \cite{holevo, schumacher, winter} for memoryless quantum channels). To be specific the classical product state capacity of a family $\{\mathcal{N}_t:\mathcal{B}(\hr')\to\mathcal{B}(\hr)  \}_{t\in T}$ of quantum channels, as described by completely positive, trace preserving maps, is given, according to our results, by
\[C_{1}(\{\mathcal{N}_t  \}_{t\in T} )=\sup_{\{p_i, D_i\}  }\inf_{t\in T}\chi (\{p_i,\mathcal{N}_t(D_i)\}),  \]
where the supremum is taken over all ensembles $\{p_i, D_i\} $ of possible input states $D_i\in\mathcal{S}(\hr')$ occurring according to probability distribution $(p_i)$, and
\[ \chi (\{p_i,\mathcal{N}_t(D_i)\}):= S\left(\sum p_i \mathcal{N}_t(D_i) \right)-\sum p_i S(\mathcal{N}_t(D_i)). \]
The full classical capacity of $\{\mathcal{N}_t  \}_{t\in T} $ is then
\[C(\{\mathcal{N}_t  \}_{t\in T} )=\lim_{n \to\infty}\frac{1}{n}C_1(\{\mathcal{N}_t^{\otimes n}  \}_{t\in T} ),  \]
and the limit is in general necessary by a counterexample to the additivity conjecture given by Hastings \cite{hastings}.\\
The capacity results for compound and averaged cq-channels show nicely the impact of the degree of channel uncertainty on the capacity. In fact, for the compound cq-channel we merely know that the information transmission happens over an unknown memoryless cq-channel which belongs to an a priori given set of channels. The capacity formula (\ref{quant-compound-capacity}) is the best worst-case rate we can guarantee simultaneously for all involved channels. For averaged cq-channels, on the other hand, the formula  (\ref{averaged-capacity}) takes into account only the almost sure worst-case cq-channel, since we are given an additional information represented by the probability measure on the memoryless branches. Consequently, the capacity of compound-cq-channels is smaller than the capacity of their averaged counterparts in many natural situations. A simple example illustrating this effect is as follows. \\
Let $T:=\{1,\ldots , K\}$ be a finite set and let $W_1,\ldots, W_K:\{0,1\}\to\mathcal{S}(\cc^2)$ be cq-channels that defined as follows. Let $W_1$ be any channel with the capacity $C(W_1)=0$. For $j\in \{2,\ldots, K  \}$ select distinct unitaries $U_2,\ldots, U_K$ acting on $\cc^2$ and define $W_j(b):=U_j|e_{b}\rangle\langle e_b|U_j^{\ast}$ where $b\in \{0,1  \}$, $j\in\{2,\ldots, K  \}$ and $e_0,e_1$ is the canonical basis of $\cc^2$. Note that for each $p\in\mathcal{P}(\{0,1  \})$ and $j\in\{2,\ldots, K  \}$
\[\chi(p,W_j)=H(p)  \]
holds, and consequently $C(W_2)=\ldots =C(W_K)=1$.
 Since any sequence of codes with asymptotically vanishing probability of error for the compound cq-channel $T$ has to be reliable for each of our channels $W_1,\ldots ,W_K$ and especially for $W_1$, we see that the only achievable rate for $T$ is $0$. Consequently $C(T)=0$. Now, if both the transmitter and receiver have additional information that the channels from $T$ are drawn according to a priori probability distribution $\mu(1)=0$ and $\mu(i)=\frac{1}{K-1}$ for $i\in\{2,\ldots, K  \}$  then it follows from Theorem \ref{direct-averaged} that
\begin{eqnarray*}C(W)&\ge& \sup_{p\in\mathcal{P}(\{0,1  \})}\mathrm{ess-}\inf_{t\in T}\chi(p,W_t)\\
                     &=& \sup_{p\in\mathcal{P}(\{0,1  \})}\min_{i\in\{2,\ldots, K\}}\chi(p,W_t)\\
                     &=& \sup_{p\in\mathcal{P}(\{0,1\})}H(p)\\
                     &=& 1,  
\end{eqnarray*}
where $W$ denotes the averaged channel associated with $T$ and $\mu$.
\section*{Acknowledgment}
We are grateful to M. Hayashi who helped us clarify the story of his approach to universal quantum hypothesis testing.\\
We thank the Associate Editor and the anonymous referee for many useful comments and suggestions that improved the readability of the paper.
\appendices
\section{Proof of Theorem \ref{one-shot-coding-theorem}}\label{proof-one-shot-coding-theorem}
This appendix is devoted to the proof of Theorem \ref{one-shot-coding-theorem}. We will apply a random coding argument of Hayashi and Nagaoka which in turn is based on the following operator inequality which we quote from the work \cite{hayashi-nagaoka} by Hayashi and Nagaoka:
\begin{theorem}[Hayashi \& Nagaoka \cite{hayashi-nagaoka}]\label{hayashi-nagaoka}
Let $\mathcal{K}$ be a finite-dimensional Hilbert space. For any operators $a,b\in\mathcal{B}(\mathcal{K})$ with $0\le a\le \idn$ and $b\ge 0$, we have
\begin{equation}\label{key-inequality}
\idn-\sqrt{a+b}^{-1}a \sqrt{a+b}^{-1}\le 2(\idn-a)+4b, 
\end{equation}
where $(\cdot)^{-1}$ denotes the generalized inverse.
\end{theorem}
Let us first note that our projection $P\in \mathcal{B}_{\textrm{diag}}\otimes \mathcal{B}(\mathcal{K})$ can be uniquely  written as 
\[P=\sum_{k\in K}|k\rangle\langle k|\otimes P_k,  \]
with suitable projections $P_k\in\mathcal{B}(\mathcal{K}) $ for all $k\in K$. With this representation we have
\begin{equation}\label{one-shot-coding-1}
  \textrm{tr}(\rho P)=\sum_{k\in K}w(k)\textrm{tr}(D_k P_k),
\end{equation}
and
\begin{equation}\label{one-shot-coding-2}
  \textrm{tr}((w\otimes \sigma)P )=\sum_{k\in K}w(k)\textrm{tr}(\sigma P_k).
\end{equation}
Now let us set $M:=[2^{\mu-\gamma} ]$ and consider i.i.d. random variables $U_1,\ldots ,U_M$ with values in $K$ each of which is distributed according to $w\in \mathcal{P}(A)$. Moreover we set
\begin{equation}\label{one-shot-coding-3}
  b_{i}(U_1,\ldots,U_M):=\left(\sum_{j=1}^{M}P_{U_j} \right)^{-1/2}P_{U_i}\left(\sum_{j=1}^{M}P_{U_j} \right)^{-1/2}.
\end{equation}
Applying Lemma \ref{hayashi-nagaoka} we obtain
\begin{equation}\label{one-shot-coding-4}
  \idn_{\mathcal{K}}-b_{i}(U_1,\ldots, U_M)\le 2(\idn_{\mathcal{K}}-P_{U_i} )+4 \sum_{\begin{subarray}{l}j=1\\ j\neq i \end{subarray}}^{M} P_{U_j}.
\end{equation}
In the following consideration we use the shorthand $e(U)$ for the average error probability of the random code $(U_i,b_{i}(U_1,\ldots, U_M))_{i=1}^{M}$, i.e. we set 
\[ e(U):=\frac{1}{M}\sum_{i=1}^{M}\textrm{tr}(D_{U_i}(\idn_{\mathcal{K}}-b_{i}(U_1,\ldots U_M)  )). \]
Recalling the fact that $U_1,\ldots ,U_M$ are i.i.d. each distributed according to $w$ and (\ref{one-shot-coding-4}) yields
\begin{eqnarray}\label{one-shot-coding-5}
  \mathbb{E}_{U_1,\ldots ,U_M}(e(U))&\le & \frac{2}{M}\sum_{i=1}^{M}\sum_{k\in K}w(k)\textrm{tr}(D_k( \idn_{\mathcal{K}}-P_k ))\nonumber\\
&&+\frac{4(M-1)M}{M}\sum_{k\in K}w(k)\textrm{tr}(\sigma P_k)\nonumber\\
&\le& 2\textrm{tr}(\rho (\idn-P))+4\cdot M \cdot \textrm{tr}((w\otimes \sigma)P)\nonumber\\
&\le& 2\cdot \lambda +4\cdot 2^{-\gamma},
\end{eqnarray}
where we have used (\ref{one-shot-coding-1}) and (\ref{one-shot-coding-2}) in the second inequality. (\ref{one-shot-coding-5}) shows that there must be at least one deterministic code $(k_i,b_i)_{i=1}^{M}$, which is a realization of the random code $(U_i,b_{i}(U_1,\ldots, U_M))_{i=1}^{M}$, with average error probability less than $ 2\cdot \lambda +4\cdot 2^{-\gamma}$ which concludes the proof of Theorem \ref{one-shot-coding-theorem}.

\end{document}